\pdfoutput=1
\documentclass[sigconf,nonacm]{acmart}
\graphicspath{{./}}
\usepackage{bm}
\usepackage{cleveref}
\usepackage{xspace}
\usepackage{float}
\usepackage{placeins}
\usepackage{enumitem}
\usepackage{booktabs}
\usepackage{tikz}
\usetikzlibrary{arrows.meta,patterns}
\definecolor{figblue}{HTML}{0072B2}
\definecolor{figorange}{HTML}{D55E00}
\definecolor{figgray}{HTML}{7F7F7F}
\newcommand{\ncalPosMinAlive}{1}
\newcommand{\ncalPosMedAlive}{136}

\newcommand{\nnDeadCalAlsoTest}{3}
\newcommand{\nnTest}{2,538,132}
\newcommand{\nnTrain}{8,248,933}
\newcommand{\nnLabels}{24}
\newcommand{\nnVal}{761,439}
\newcommand{\nwMax}{8,248,932}
\newcommand{\nwMin}{78}
\newcommand{\npopMAP}{0.535}
\newcommand{\nSraw}{0.117}
\newcommand{\nSplIsoId}{0.465}
\newcommand{\nSplIsoPrior}{0.784}

\newcommand{\nSpooled}{0.261}
\newcommand{\nSceil}{0.786}
\newcommand{\nSauc}{0.793}
\newcommand{\nSsatOne}{16.6}

\newcommand{\nSnDead}{3}
\newcommand{\nWraw}{0.776}
\newcommand{\nWplIsoId}{0.434}
\newcommand{\nWplIsoPrior}{0.828}

\newcommand{\nNraw}{0.808}

\newcommand{\nNplIsoPrior}{0.819}

\newcommand{\nNauc}{0.824}

\newcommand{\nSelkan}{0.151}
\newcommand{\nSloss}{0.690}
\newcommand{\ntieShare}{3}
\newcommand{\nrepairShare}{97}

\newcommand{\nelkanShare}{5}
\newcommand{\nSceilGap}{0.021}
\newcommand{\nSplIsoPriorCeilGap}{0.003}
\newcommand{\nnCappedShift}{19}
\newcommand{\nSplOffsetPrior}{0.738}

\newcommand{\noffsetShare}{90}

\newcommand{\nWdeadAllRowsPct}{100}
\newcommand{\nWdeadRawMean}{0.120}
\newcommand{\nWaliveCalMax}{0.0106}

\newcommand{\nfracRowsWithPos}{2.9}
\newcommand{\nnEvalRows}{1,776,693}

\newcommand{\nSsqrtRaw}{0.688}
\newcommand{\nStenrRaw}{0.036}
\newcommand{\nSsqrtPlIsoPrior}{0.806}
\newcommand{\nStenrPlIsoPrior}{0.772}

\newcommand{\nwhatifIdeal}{0.646}

\newcommand{\noddsShare}{23}

\newcommand{\noverlapWhatif}{3.11}
\newcommand{\noverlapS}{2.57}

\newcommand{\nbootB}{2000}
\newcommand{\nbootN}{51,403}

\newcommand{\nciWprior}{+0.052~[+0.050,\,+0.054]}
\newcommand{\nciNprior}{+0.012~[+0.011,\,+0.013]}
\newcommand{\nciNminusS}{+0.024~[+0.022,\,+0.026]}

\newcommand{\nnMulan}{11}
\newcommand{\nnLearners}{5}

\newcommand{\nnMulanMaxrMax}{3149}
\newcommand{\nenronStrongRawWone}{0.653}

\newcommand{\nenronStrongInvWr}{0.656}

\newcommand{\nmedStrongDistinctWone}{24}
\newcommand{\nmedStrongDistinctWtenr}{2.7}

\newcommand{\nmedStrongSatWtenr}{0.00}
\newcommand{\nnMulanModest}{4}
\newcommand{\nmodestDefaultWorst}{-0.041}
\newcommand{\nmodestDefaultBest}{-0.001}
\newcommand{\nmodestAnyWorst}{-0.146}

\newcommand{\nmodestAnyBest}{+0.010}

\newcommand{\nbirdsRawWone}{0.593}
\newcommand{\nbirdsPlIsoPriorWone}{0.549}

\newcommand{\nmedRawWone}{0.683}
\newcommand{\nmedPlIsoPriorWone}{0.651}

\newcommand{\nnHurtDefault}{3}

\newcommand{\ncalPosMedBirds}{5}
\newcommand{\ncalPosMedGenbase}{5}
\newcommand{\ncalPosMedMedical}{1}
\newcommand{\ncalPosMedEnron}{6}

\newcommand{\nnWhatifClose}{8}

\newcommand{\ngenbaseDefaultWhatifWr}{0.927}
\newcommand{\ngenbaseDefaultRawWr}{0.868}
\newcommand{\nnDoseCells}{55}
\newcommand{\nnCollapseCells}{8}
\newcommand{\nnCollapseCellsAtPointThree}{4}
\newcommand{\nnCollapseCellsAtPointSeven}{10}
\newcommand{\nnCollapseCellsAtPointNine}{16}
\newcommand{\ncollapseDatasets}{corel5k, delicious, genbase}
\newcommand{\nnDelCollapseLearners}{5}
\newcommand{\nnCorelCollapseLearners}{2}
\newcommand{\nnOtherCollapseCells}{1}
\newcommand{\ncorelL}{374}

\newcommand{\ncorelNtestDose}{500}

\newcommand{\ncorelDeadCalDose}{39}
\newcommand{\ncorelRawWone}{0.236}
\newcommand{\ncorelRawWr}{0.000}

\newcommand{\ncorelPlIsoPriorWr}{0.244}
\newcommand{\ncorelPlIsoIdWr}{0.060}
\newcommand{\ncorelPooledWr}{0.125}
\newcommand{\ncorelInvWr}{0.001}

\newcommand{\ndelL}{983}

\newcommand{\ndelNtestDose}{3,185}

\newcommand{\ndelRawWone}{0.419}
\newcommand{\ndelRawWr}{0.001}

\newcommand{\ndelPlIsoPriorWr}{0.434}

\newcommand{\ndelPooledWr}{0.025}
\newcommand{\ndelInvWr}{0.003}

\newcommand{\ndelAucWone}{0.752}
\newcommand{\ndelAucWr}{0.705}

\newcommand{\nenronMaxr}{786}
\newcommand{\nbibtexL}{159}

\newcommand{\nmaxRepairGapCollapse}{0.028}
\newcommand{\nminPooledGapCollapse}{0.053}
\newcommand{\nnRepairCollapseCells}{10}

\newcommand{\ncorelIdGapMin}{0.000}

\newcommand{\ntauList}{0, 1, 2, 5, 10}
\newcommand{\ntauFolds}{5}

\newcommand{\ndelCiTauSel}{+0.016~[+0.012,\,+0.020]}

\newcommand{\nnHurtTauZeroCI}{2}
\newcommand{\nnHurtAfterTau}{0}

\newcommand{\ntauWorstSelMinusRaw}{-0.134}

\newcommand{\nnTauCells}{110}

\newcommand{\nnTauChoicePerLabel}{45}
\newcommand{\ntauMax}{10}
\newcommand{\nnTauChoiceMaxTau}{23}

\newcommand{\nnSharedBelowRaw}{9}

\newcommand{\nnTauSelBelowRaw}{6}
\newcommand{\nnTauZeroBelowRaw}{29}

\newcommand{\ndepSraw}{0.117}

\newcommand{\ndepSplIsoPrior}{0.780}

\newcommand{\nnValRows}{1,903,599}

\newcommand{\nprotoMaxDiff}{0.004}

\newcommand{\nnSurveyImpl}{17}
\newcommand{\nInLabels}{4,000}

\newcommand{\nInVal}{39,362}
\newcommand{\nIwMax}{2,301}
\newcommand{\nIwMin}{6}
\newcommand{\nIlnwMin}{1.8}
\newcommand{\nIlnwMax}{7.7}
\newcommand{\nIpopMAP}{0.056}

\newcommand{\nIcalPosMedAlive}{37}

\newcommand{\nISraw}{0.001}

\newcommand{\nISplIsoPrior}{0.225}

\newcommand{\nISceil}{0.230}

\newcommand{\nISnDead}{0}

\newcommand{\nINraw}{0.131}

\newcommand{\nINplIsoPrior}{0.250}

\newcommand{\nINceil}{0.269}

\newcommand{\nISloss}{0.129}

\newcommand{\nIoddsShare}{32}

\newcommand{\nIoverlapS}{0.02}
\newcommand{\nIfracRowsChangedWhatif}{79}

\newcommand{\nIciSpriorMinusElkan}{+0.208~[+0.206,\,+0.210]}

\newcommand{\nIprotoMaxDiffPair}{0.007}

\newcommand{\nMSplIsoPrior}{0.549}

\newcommand{\nMSceil}{0.548}

\newcommand{\nMSsatOne}{23.3}

\newcommand{\nMNraw}{0.766}

\newcommand{\nMtieShare}{33}
\newcommand{\nMelkanShare}{21}

\newcommand{\nMwhatifIdeal}{0.127}

\newcommand{\nModdsShare}{98}

\newcommand{\nIMSraw}{0.008}

\newcommand{\nIMSplIsoPrior}{0.081}

\newcommand{\nIMSauc}{0.836}

\newcommand{\nIMNraw}{0.056}

\newcommand{\nIMNauc}{0.424}

\newcommand{\nIMSelkan}{0.070}

\newcommand{\ncapNoneSatOne}{16.6}

\newcommand{\ncapLnwMax}{15.9}

\newcommand{\nSmdsPointThreeraw}{0.634}

\newcommand{\nSmdsPointThreeplIsoPrior}{0.556}
\newcommand{\nSmdsPointThreebPriorMax}{0.90}

\newcommand{\nSmdsPointThreeTEtaC}{0.9}
\newcommand{\nSmdsPointThreeoverMedian}{7.0}

\newcommand{\nSmdsFiveraw}{0.718}

\newcommand{\nSmdsFiveelkan}{0.806}
\newcommand{\nSmdsFiveplIsoPrior}{0.829}
\newcommand{\nSmdsFivebPriorMax}{8.90}

\newcommand{\nSmdsFiveTEtaC}{15.0}
\newcommand{\nSmdsFiveoverMedian}{1.1}

\newcommand{\nSmdsFiveNraw}{0.823}

\newcommand{\nSmdsFiveNplIsoPrior}{0.829}

\newcommand{\ncapN}{5}
\newcommand{\ncapNwithinBound}{5}

\newcommand{\ncapT}{60}
\newcommand{\ncapEta}{0.05}

\newcommand{\nSSretrainRaw}{0.111}

\newcommand{\nSNretrainRaw}{0.808}

\newcommand{\ncapTwinGapMax}{+0.198}
\newcommand{\ncapNtwinsAhead}{4}

\newcommand{\ntieShareUnw}{5}
\newcommand{\ntieShareOracle}{8}

\newcommand{\nseedN}{3}
\newcommand{\nseedSraw}{0.092 \pm 0.052}

\newcommand{\nseedSplIsoPrior}{0.788 \pm 0.003}

\newcommand{\nseedOddsShare}{20.3 \pm 2.3}
\newcommand{\nIcapNoneSatOne}{19.4}

\newcommand{\nIcapLnwMax}{7.7}

\newcommand{\nIStieCellsExactOne}{19.4}

\newcommand{\nIStieCellsPerRowMean}{774.27}

\newcommand{\nItieShareUnw}{54}

\newcommand{\nIMAnLabels}{49,688}

\newcommand{\nIMApopMAP}{0.051}

\newcommand{\nIMASplIsoPrior}{0.073}
\newcommand{\nIMASplIsoId}{0.001}

\newcommand{\nIMASnDead}{21,136}

\newcommand{\nSbOverLnwMedian}{0.98}
\newcommand{\nSnIdentShift}{5}

\newcommand{\nSbDirectOverLnwMedian}{0.78}

\newcommand{\nSmdsPointThreeCeil}{0.555}

\newcommand{\nnCapsIsoBest}{3}
\newcommand{\nnCapsRungs}{5}

\newcommand{\nMSnIdentShift}{1}

\newcommand{\nISbOverLnwMedian}{0.88}
\newcommand{\nISnIdentShift}{394}

\newcommand{\nISbDirectOverLnwMedian}{0.50}

\newcommand{\nInCapsIsoBest}{5}

\newcommand{\nIMSbOverLnwMedian}{1.06}
\newcommand{\nIMSnIdentShift}{3644}

\newcommand{\nsatMarginF}{17.3}

\newtheorem{proposition}{Proposition}
\newtheorem{lemma}{Lemma}
\newtheorem{corollary}{Corollary}
\newtheorem{remark}{Remark}
\crefname{proposition}{Proposition}{Propositions}
\Crefname{proposition}{Proposition}{Propositions}
\crefname{lemma}{Lemma}{Lemmas}
\Crefname{lemma}{Lemma}{Lemmas}
\crefname{corollary}{Corollary}{Corollaries}
\Crefname{corollary}{Corollary}{Corollaries}
\crefname{remark}{Remark}{Remarks}
\Crefname{remark}{Remark}{Remarks}
\crefname{section}{Section}{Sections}
\Crefname{section}{Section}{Sections}
\crefname{table}{Table}{Tables}
\Crefname{table}{Table}{Tables}
\crefname{figure}{Figure}{Figures}
\Crefname{figure}{Figure}{Figures}
\crefname{appendix}{Appendix}{Appendices}
\Crefname{appendix}{Appendix}{Appendices}

\newcommand{\Sm}{\textsf{LGBM-w}\xspace}
\newcommand{\Wm}{\textsf{WGB}\xspace}
\newcommand{\Nm}{\textsf{LGBM-0}\xspace}
\newcommand{\mapk}{\mathrm{MAP}@K\xspace}
\newcommand{\mapseven}{\mathrm{MAP}@7\xspace}
\newcommand{\Supp}[1]{Supplement~#1}  

\setcopyright{none}
\renewcommand\footnotetextcopyrightpermission[1]{}
\begin{document}

\title{Odds-Shift Slippage in One-vs-Rest Rankers: Diagnosing and Repairing Reweighting-Induced Top-$K$ Errors}

\author{Akifumi Goto}
\orcid{0009-0004-8305-6436}
\email{s6025131@st.shiga-u.ac.jp}
\affiliation{%
  \institution{Graduate School of Data Science, Shiga University}
  \city{Hikone}
  \state{Shiga}
  \country{Japan}
}

\begin{abstract}
One-vs-rest rankers that show each user the top-$K$ of many rare labels usually counter imbalance with a per-label positive-class weight, \texttt{scale\_pos\_weight}$=n_-/n_+$. Elkan's identity says such a weight shifts label $j$'s log-odds by $\ln w_j$, so the model ranks by weighted odds rather than by the marginal that is Bayes-optimal for precision@$K$, and it suggests inverting the shift afterwards; what a finite learner does with a weight in the thousands, and which repair then works, has not been measured. We call the gap between the promised shift and the realized one \emph{odds-shift slippage} and measure it on matched pairs of LightGBM and MLP models that differ only in the weights. On Santander the weight takes $\mapseven$ from $\nNraw$ to $\nSraw$; for the boosted pairs the ideal odds shift accounts for \noddsShare\% of that loss (\nIoddsShare\% on Instacart; \nModdsShare\% for an MLP pair on the same rows) and slippage for the rest. We prove that a booster whose leaf steps are capped at $c$ realizes at most $T\eta c$ nat of shift in $T$ rounds at rate $\eta$, which a cap sweep confirms, and show that without a cap saturated cells tie at exactly $1.0$, beyond the reach of any separable map. The analytic inversion therefore pays only where the shift was realized and nothing saturated, whereas per-label isotonic regression returns the Santander model to $\nSplIsoPrior$ ($\ndepSplIsoPrior$ with the calibrator fitted on the validation period) --- but only if labels without calibration positives are mapped to their prior rather than passed through. On \nnMulan{} public MULAN benchmarks and \nnLearners{} learners the weighted model loses more than half of its $\mapk$ in \nnCollapseCells{} of \nnDoseCells{} cells, and on delicious and Corel5k the same repair returns it to the unweighted level; per-label calibration hurts where positives are scarce, a harm that a cross-validated rule removes for the default learner. The recipe --- estimate marginals unweighted, calibrate per label with an explicit prior fallback, choose the calibrator on the calibration split, then decide at top-$K$ --- is released as \texttt{oddslip}.
\end{abstract}

\maketitle

\section{Introduction}
\label{sec:intro}

Ranking many binary labels within a row is the operational core of product recommendation, tagging and risk triage: for each customer the system scores every product, sorts, and shows the top $K$~\cite{Jain2016XMLLoss,Menon2019Reductions}. The labels are almost always imbalanced, and the most common remedy in gradient-boosting practice is a per-label positive-class weight, \texttt{scale\_pos\_weight} in XGBoost and LightGBM~\cite{Chen2016XGBoost,Ke2017LightGBM}, for which the XGBoost documentation suggests $n_-/n_+$~\cite{XGBoostDocs}; neural multi-label models use the same recipe as \texttt{pos\_weight}.

Label-specific weights are not a free choice when the deliverable is a per-row ranking. Ranking labels by their marginal $p_{ij}=\Pr(y_{ij}=1\mid x_i)$ is Bayes-optimal for precision@$K$, and one-vs-rest training with a proper loss is consistent for it~\cite{Menon2019Reductions,Wydmuch2018PLT}. A positive-class weight $w$ moves the population minimizer of a proper loss from $p$ to $wp/(wp+1-p)$, an odds shift by $w$~\cite{Elkan2001,Saerens2002}. Together the two facts say that a label-specific $w_j$ makes the model rank by $w_j\cdot\mathrm{odds}(p_{ij})$, a different ranking whenever the weights differ.

What the classical account leaves open is what a finite learner does when $w_j$ is in the thousands, as the recommended value is for rare labels, and which post-hoc repair then works. Independent per-class calibration maps are known to lose top-$k$ accuracy in multi-class settings~\cite{Patel2021IMax}, and the extreme multi-label literature deliberately uses one shared monotone map, deferring per-label calibration to future work~\cite{Ullah2025XMLCCalib}. Neither measures the cross-label consequence of training weights on a top-$K$ benchmark, nor what a calibrator does with labels that have no calibration positives.

\Cref{fig:one} gives our answer on two product-recommendation benchmarks, Santander and Instacart, each on a pair of LightGBM models that differ only in the weights. Applying the ideal odds shift to the unweighted model reproduces only a minority of the loss the trained weighted model actually suffers. We call the remainder \emph{odds-shift slippage}: the shift the learner realizes is not the shift the weight asked for. \Cref{fig:schematic} sketches the two regimes that produce it. A leaf-step cap keeps the realized shift below $\ln w_j$ (\cref{prop:budget}); without a cap the learner realizes the shift on the labels where it is still defined, while elsewhere cells saturate into ties that no per-label map can reorder (\cref{lem:leaf}). Among the configurations we evaluated, analytic inversion improves only those that combine a near-intended realized shift with no observed saturation (\cref{fig:regime}); per-label calibration repairs the ordered part, provided that a label without calibration positives is mapped to its prior rather than passed through (\cref{prop:dead}).

\paragraph{Contributions}
\begin{enumerate}[leftmargin=*,nosep]
\item \textbf{Odds-shift slippage, decomposed.} Using Elkan's identity as a what-if device, we split the top-$K$ loss of a weighted one-vs-rest ranker into the ideal odds shift and the loss beyond it, the slippage, which itself splits into the component that a per-label monotone map fitted on calibration data recovers and a residue above the in-sample per-label ceiling (\cref{fig:one}A, \cref{tab:ladder}). For the boosted pairs the shift accounts for \noddsShare\% of the observed loss on Santander ($\nseedOddsShare\%$ over three train-row draws) and \nIoddsShare\% on Instacart; the share is a property of the pair, since an MLP pair on the same Santander rows gives up \nModdsShare\% of its loss to the same ideal shift. Logit adjustment and residual reranking~\cite{Menon2021LogitAdjust,Wang2026BeyondLA} correct a trained score; neither measures how far a training weight's intended shift was realized, which is what the decomposition isolates.
\item \textbf{A step budget that makes the account falsifiable.} \Cref{prop:budget} bounds the realized per-label shift of a booster with leaf-step cap $c$ by $T\eta c$ whatever the weight, and \cref{cor:exact} says when the analytic inversion is exact. A sweep of \ncapN{} caps on both benchmarks keeps every realized maximum inside its budget with no saturation (\cref{fig:one}B); without a cap the shift is identifiable only on labels with no saturated cell, where its median is close to $\ln w_j$. Among the evaluated models the inversion pays only for the two that realized most of the shift and saturated nothing (\cref{fig:regime}), an association rather than a proven condition; and an unweighted twin trained with the same cap matches the calibrated capped weighted arm, so on Santander, the one dataset with such twins, the cap, not the weight, is what helps; we know of no earlier statement of when the analytic inversion~\cite{Elkan2001} can be exact for a booster under a leaf-step cap.
\item \textbf{Per-label calibration with a prior fallback is the repair that travels, with its selection rule.} Per-label isotonic regression with dead labels mapped to their prior returns the weighted Santander booster to within $\nSplIsoPriorCeilGap$ of its in-sample ceiling under the test-split protocol and within $0.01$ under the deployment protocol; a shared map repairs nothing beyond tie resolution (\cref{rem:shared}), and the identity fallback turns the same calibrator into a loss (\cref{prop:dead}). The collapse and its repair reproduce on \nnMulan{} MULAN benchmarks and \nnLearners{} learners (\cref{fig:mulan}), where the calibrator also \emph{hurts} an unweighted model with scarce calibration positives; a cross-validated choice inside the calibration split removes that harm for the default learner and fails in \nnTauSelBelowRaw{} of \nnTauCells{} cells. Per-label monotone maps are known~\cite{Zadrozny2002,Tae2026MRP}; what is new is the source of the distortion, the fallback that decides the sign, and the selection by the deployed metric.
\end{enumerate}

\section{Setup and Theory}
\label{sec:theory}

\begin{figure*}[t]
\centering
\begin{tikzpicture}[x=1cm,y=1cm,font=\scriptsize,
  every node/.style={inner sep=1pt},
  axis/.style={-{Latex[length=1.2mm]},figgray!80!black,line width=0.4pt}]
\begin{scope}
  \node[anchor=west,font=\scriptsize\bfseries] at (-0.2,3.55) {(a) Two slippage regimes};
  \fill[figorange!15] (2.9,0) rectangle (4.6,3.15);
  \node[figorange!85!black,anchor=north,align=center,text width=1.6cm] at (3.75,1.72)
    {no cap: cells at $1.0$, $b_j$ undefined (Lemma~1)};
  \draw[axis] (0,0) -- (4.8,0);
  \node[anchor=north east,figgray!80!black] at (4.8,-0.05) {intended shift $\ln w_j$};
  \draw[axis] (0,0) -- (0,3.3);
  \node[anchor=south west,figgray!80!black] at (0.05,3.02) {realized shift $b_j$};
  \draw[figblue,dashed,line width=0.5pt] (0,0) -- (2.9,2.9);
  \node[figblue,anchor=south east] at (2.55,2.62) {$b_j=\ln w_j$ (Prop.~1)};
  \draw[figorange,line width=0.9pt] (0,0) -- (2.9,2.9);
  \draw[figorange,line width=0.9pt] (0,0) -- (1.45,1.45) .. controls (1.75,1.75) and (1.85,1.85) .. (4.6,1.88);
  \draw[figgray,dotted,line width=0.5pt] (0,1.88) -- (4.6,1.88);
  \node[figgray!80!black,anchor=east] at (-0.05,1.88) {$T\eta c$};
  \node[figorange!85!black,anchor=south west,align=left] at (3.0,1.98) {cap: shortfall,\\$b_j\le T\eta c$ (Prop.~2)};
  \node[anchor=north west,align=left,text width=4.7cm,figgray!60!black] at (-0.15,-0.42)
    {Inverting by $\ln w_j$ is exact only on the dashed line: a capped learner falls short of it, an uncapped one leaves it.};
\end{scope}
\begin{scope}[xshift=6.3cm]
  \node[anchor=west,font=\scriptsize\bfseries] at (-0.2,3.55) {(b) Which repair returns which part};
  \def\xa{1.9}\def\xb{2.7}\def\xc{4.1}\def\xd{4.75}\def\yb{2.15}
  \fill[figgray!55] (\xa,\yb) rectangle (\xb,\yb+0.45);
  \fill[figorange!45] (\xb,\yb) rectangle (\xc,\yb+0.45);
  \fill[pattern=north east lines,pattern color=figorange] (\xc,\yb) rectangle (\xd,\yb+0.45);
  \draw[figgray] (\xa,\yb) rectangle (\xd,\yb+0.45);
  \node[anchor=east,figgray!60!black] at (\xa-0.08,\yb+0.22) {the top-$K$ loss};
  \node[anchor=south] at (2.3,\yb+0.5) {odds shift};
  \node[anchor=south] at (3.4,\yb+0.82) {slippage (per-label, monotone)};
  \node[anchor=south] at (4.42,\yb+0.5) {ties};
  \draw[figgray!70,line width=0.3pt] (3.4,\yb+0.48) -- (3.4,\yb+0.8);
  \node[anchor=east] at (\xa-0.08,1.6) {inversion by $\ln w_j$};
  \draw[figgray!50,line width=0.4pt] (\xa,1.6) -- (\xd,1.6);
  \draw[figblue,line width=2.2pt] (\xa,1.6) -- (\xb,1.6);
  \node[anchor=east] at (\xa-0.08,1.1) {per-label monotone map};
  \draw[figgray!50,line width=0.4pt] (\xa,1.1) -- (\xd,1.1);
  \draw[figblue,line width=2.2pt] (\xa,1.1) -- (\xc,1.1);
  \node[anchor=east] at (\xa-0.08,0.6) {shared (pooled) map};
  \draw[figgray!50,line width=0.4pt] (\xa,0.6) -- (\xd,0.6);
  \node[anchor=north west,align=left,text width=5.0cm,figgray!60!black] at (-0.15,-0.42)
    {Blue: the part of the loss the repair can return. The inversion is exact only if $b_j=\ln w_j$ (Cor.~2); a shared map changes nothing beyond tie-breaking (Rem.~1); no separable map reaches the ties.};
\end{scope}
\begin{scope}[xshift=12.4cm]
  \node[anchor=west,font=\scriptsize\bfseries] at (-0.2,3.55) {(c) Dead-label takeover (Rem.~2)};
  \draw[axis] (0,0) -- (0,3.3);
  \node[anchor=south,figgray!80!black,rotate=90] at (-0.12,1.55) {score within one row};
  \draw[figgray,dotted] (0,1.25) -- (4.4,1.25);
  \node[anchor=west,figgray!80!black] at (4.42,1.25) {$\bar q$};
  \foreach \x/\h in {0.45/1.1, 0.85/0.85, 1.25/0.65, 1.65/0.42, 2.05/0.28}
    \draw[figblue,line width=2.2pt] (\x,0) -- (\x,\h);
  \node[figblue,anchor=south,align=center] at (1.25,1.32) {calibrated labels\\$g_j(s)\le\bar q$};
  \draw[figorange,line width=2.2pt] (3.05,0) -- (3.05,2.55);
  \node[figorange!85!black,anchor=south,align=center] at (3.05,2.6) {dead label,\\identity fallback};
  \node[figorange!85!black,anchor=west,align=left] at (3.15,1.9) {takes a\\top-$K$ slot};
  \draw[figorange!70!black,line width=2.2pt] (3.95,0) -- (3.95,0.2);
  \node[figorange!70!black,anchor=south,align=center] at (3.95,0.3) {prior\\fallback\\$\hat\pi_d$};
  \node[anchor=north west,align=left,text width=4.6cm,figgray!60!black] at (-0.15,-0.42)
    {A label with no calibration positives has no map. Passed through, its raw score outranks every calibrated label; mapped to its prior, it does not.};
\end{scope}
\end{tikzpicture}
\Description{Three schematic panels: realized against intended shift with a shortfall regime under a cap and a saturation regime without one; a bar of the top-K loss split into odds shift, slippage and ties with the reach of each repair; and a row of scores in which a dead label passed through the identity fallback outranks every calibrated label.}
\caption{\textbf{Schematic (no data).} \textbf{(a)} Elkan's identity puts the realized per-label shift $b_j$ on the diagonal. A booster with leaf-step cap $c$ cannot exceed $T\eta c$ (\cref{prop:budget}); without a cap the cells of rare labels saturate at a stored probability of exactly $1.0$, where $b_j$ is not identifiable (\cref{lem:leaf}). \textbf{(b)} The inversion returns the odds-shift part only where $b_j=\ln w_j$ (\cref{cor:exact}); a per-label monotone map also returns the slippage; nothing separable reorders the ties (\cref{rem:shared}). \textbf{(c)} A label with no calibration positives has no map; passed through, its raw score outranks every calibrated label of the row (\cref{prop:dead}).}
\label{fig:schematic}
\end{figure*}
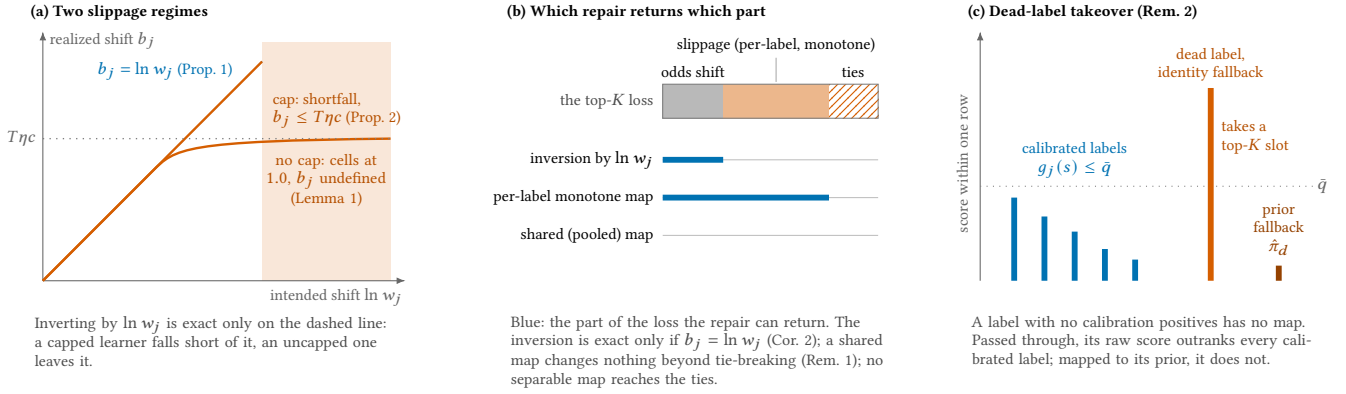

Rows $i=1,\dots,n$ carry features $x_i$ and binary labels $y_{ij}$, $j=1,\dots,L$; $p_{ij}=\Pr(y_{ij}=1\mid x_i)$ is the marginal and $\pi_j=\Pr(y_{ij}=1)$ the prevalence. A scorer returns $s_{ij}$ and the system shows the $K$ largest scores of row $i$; $\sigma$ is the logistic function and shifts are in nat. We evaluate with $\mapk$ in the Kaggle definition (average precision at $K$ per row, averaged over rows with at least one positive) and report per-row bootstrap intervals. For precision@$K$ the population-optimal ranking sorts by $p_{ij}$, and one-vs-rest with a proper loss is consistent for it~\cite{Koyejo2015Consistent,Menon2019Reductions}. $\mapk$, the benchmark metric of both datasets, is a different objective (it normalizes per row by the number of positives); every optimality claim in this paper is about precision@$K$, while the ordering statements below (\cref{prop:odds,cor:exact,rem:shared,prop:dead}) concern within-row order and hold for any metric of the top-$K$ order, $\mapk$ included. A post-hoc calibrator is \emph{separable} if it maps label $j$ by a non-decreasing $g_j$ independently of the other labels; per-label Platt, isotonic, temperature and beta calibration~\cite{Platt1999,Zadrozny2002,Guo2017Calibration,Kull2017Beta} are separable, and so is the shared (pooled) isotonic map of~\cite{Ullah2025XMLCCalib} with $g_j\equiv g$. Calibration here always means the per-label marginal; the label-ranking calibration of~\cite{Thies2026LabelRankingCalib} concerns distributions over orderings and is a different object. A positive-class weight $w_j>0$ multiplies the loss of the positives of label $j$ (\texttt{scale\_pos\_weight}, \texttt{pos\_weight}). A \emph{cell} is one (row, label) score; a \emph{dead} label has no calibration positives; the two members of a matched pair are its \emph{arms} or \emph{twins}, and a \emph{rung} is one row of the repair ladder of \cref{tab:ladder}.

\begin{remark}[Shared maps never invert, distinct maps can]
\label{rem:shared}
If $g_1=\dots=g_L$ the calibrator never creates a strict inversion within a row (monotonicity), so a shared map is rank-neutral up to ties and can repair nothing beyond tie resolution. If two continuous maps differ at an interior point $a$, say $g_j(a)<g_k(a)$, continuity gives $b>a$ with $g_j(b)<g_k(a)$, and the row with $s_{ij}=b>a=s_{ik}$ is inverted. A per-label map can reorder, and in the population limit it should, since calibrated per-label maps return the marginals; it hurts only through estimation error, whose boundary case is a label without calibration positives (\cref{prop:dead}).
\end{remark}

\begin{proposition}[Odds shift, after Elkan]
\label{prop:odds}
Let $q_{ij}$ be the population minimizer of a strictly proper binary loss with positive-class weight $w_j$. Then $\operatorname{logit}q_{ij}=\operatorname{logit}p_{ij}+\ln w_j$~\cite{Elkan2001,Saerens2002}, and for $p_{ij}>p_{ik}$ the weighted scores reverse the pair iff $\operatorname{logit}p_{ij}-\operatorname{logit}p_{ik}<\ln w_k-\ln w_j$.
\end{proposition}
\begin{proof}
The weighted loss is the unweighted loss under the tilted distribution whose posterior odds are $w_j$ times the original; the flip condition is the difference of the shifted logits.
\end{proof}
\begin{corollary}
\label{cor:odds}
(a) A common weight $w_j\equiv w$ is a shared map and never changes within-row order. (b) The flipped pairs, hence the top-$K$ loss of the \emph{ideal} weighted model, depend only on the unweighted logit gaps and the intended weights, so they can be computed from an unweighted model before weighted training (\emph{what-if}); this describes a learner that realizes the shift exactly, not necessarily the one a finite learner produces. (c) When $w_jp_{ij}\ll1$, $q_{ij}\approx w_jp_{ij}$, so weighted training approximates the optimal ranking for a propensity-weighted precision@$K$ with propensities $1/w_j$~\cite{Jain2016XMLLoss}: training weights silently change the metric.
\end{corollary}

\begin{lemma}[Leaf saturation]
\label{lem:leaf}
For weighted log-loss the optimal constant in a tree leaf with $n_+$ positives and $n_-$ negatives is $\rho(w)=wn_+/(wn_++n_-)$. Any leaf with $n_+\ge1$ satisfies $\rho(w)\ge1-n_-/(wn_+)$, so its prediction is within $\varepsilon$ of one once $w\ge n_-/(\varepsilon n_+)$; all test cells routed to such leaves tie near one and lose their order across those leaves.
\end{lemma}
\begin{proof}
$-wn_+\ln\rho-n_-\ln(1-\rho)$ is minimized at $\rho=wn_+/(wn_++n_-)$, and $1-\rho=n_-/(wn_++n_-)\le n_-/(wn_+)$.
\end{proof}
The lemma is exact for a single tree whose leaf values are weighted class frequencies and describes the fixed point that a boosted ensemble approaches when it fits the weighted objective well; it is the tree analogue of importance weights losing their effect on a flexible learner~\cite{Byrd2019}. Two consequences must be kept apart. The realized shift is no longer the constant $\ln w_j$, so the analytic inversion $q\mapsto q/(q+w(1-q))$ stops working; within a single tree the distortion is monotone in each label, and how much of a boosted ensemble's distortion a per-label monotone map fitted on calibration data recovers is what the ladder of \cref{sec:santander} measures. And the cells tied at one have lost their mutual order: a separable map of the stored probabilities can order them only by label, identically in every row; this part is bounded by the gap between the in-sample per-label ceiling and the unweighted model. Coarse leaves on small data tie cells without saturation, a different mechanism (\cref{sec:dose}(iv)).

\begin{proposition}[Step budget bounds the realized shift]
\label{prop:budget}
Fix label $j$. Let the weighted and the unweighted booster start from the same initial score $F_0=\operatorname{logit}\hat\pi_j$, and let each of $T$ rounds add a tree whose leaf values are clipped to $[-c,c]$ and scaled by $\eta$ (\texttt{max\_delta\_step}$=c$). Then $|F^{w}_j(x)-\operatorname{logit}\hat\pi_j|\le T\eta c$ and $|F^{0}_j(x)-\operatorname{logit}\hat\pi_j|\le T\eta c$ for every $x$, hence $|F^{w}_j(x)-F^{0}_j(x)|\le 2T\eta c$; and the label-free prevalence-matching shift $b_j$, defined by $\frac1n\sum_i\sigma(F^{w}_j(x_i)-b_j)=\hat\pi_j$, satisfies $|b_j|\le T\eta c$. Consequently, if $T\eta c<\ln w_j$ the intended shift is unrealizable, and the analytic inversion $\operatorname{logit}q-\ln w_j$ overshoots the label-free shift by at least $\ln w_j-T\eta c$.
\end{proposition}
\begin{proof}
Each round adds at most $\eta c$ in absolute value, so the first two bounds telescope. The mean of $\sigma(F^{w}_j(x_i)-b)$ over $i$ is decreasing in $b$; at $b=T\eta c$ every argument is at most $\operatorname{logit}\hat\pi_j$, so the mean is at most $\hat\pi_j$, and at $b=-T\eta c$ it is at least $\hat\pi_j$; the solution therefore lies in $[-T\eta c,T\eta c]$.
\end{proof}
\begin{corollary}[When the inversion is exact]
\label{cor:exact}
Inverting by $\ln w_j$ returns the unweighted ranking if $F^{w}_j(x)-F^{0}_j(x)=\ln w_j$ for all $j$ and $x$, and more generally whenever the residual $\delta_j(x)=F^{w}_j(x)-F^{0}_j(x)-\ln w_j$ does not depend on $j$ within a row. Under a cap, pointwise equality requires $\ln w_j\le 2T\eta c$ for every $j$; without a cap the bound is vacuous and \cref{lem:leaf} applies: once $w_j\ge n_-/(\varepsilon n_+)$ in a leaf, its cells sit within $\varepsilon$ of one, their stored single-precision probabilities round to $1.0$ once the margin exceeds the float32 resolution, and a separable map of the stored probabilities can order them only by label index, the same in every row.
\end{corollary}
\begin{proof}
The inverted score of label $j$ at $x$ is $F^{0}_j(x)+\delta_j(x)$; a $\delta$ that does not depend on $j$ shifts every label of the row equally. Under a cap, $|F^{w}_j(x)-F^{0}_j(x)|\le 2T\eta c$ by \cref{prop:budget}, so pointwise equality to $\ln w_j$ needs $\ln w_j\le 2T\eta c$; a row-constant residual only bounds the spread $\max_j\ln w_j-\min_j\ln w_j$ by $4T\eta c$, which is why a common weight (\cref{cor:odds}(a)) escapes the bound.
\end{proof}
LightGBM starts both boosters from the unweighted label mean, which the released EDA file verifies on a one-round model.

We call the gap between the shift a weight promises, $\ln w_j$, and the shift the learner realizes \emph{odds-shift slippage}; two regimes produce it (\cref{fig:schematic}a). Under a binding cap, $b_j\le T\eta c<\ln w_j$: the realized shift falls short, no cell saturates, and the inversion over-corrects by at least the shortfall. Without a cap leaves saturate, and on a saturated label $b_j$ is not identifiable from the stored probabilities, since no finite shift moves a cell that has rounded to $1.0$ and the bisection returns its clipping bound; the inversion is left with ties beyond any separable map of the stored probabilities. Where no cell saturates the shift is identified and, as \cref{sec:santander} measures, close to $\ln w_j$. A monotone map fitted per label absorbs the ordered part of both regimes; nothing separable absorbs the ties.

\begin{remark}[Dead-label takeover]
\label{prop:dead}
Let label $d$ have no calibration positives and keep its raw score $s_{id}$ (identity fallback), while every calibrated label $j$ is mapped by $g_j$ with $g_j(s)\le\bar q$ for all $s$ (isotonic outputs are bounded by the largest block mean of the calibration split). If $s_{id}>\bar q$ then $d$ is ranked above every calibrated label in row $i$ regardless of $y_{id}$; with $m$ such labels, $\min(m,K)$ of the top-$K$ slots are occupied by them (\cref{fig:schematic}c). The condition holds whenever raw scores are inflated relative to calibrated ones, which is the situation that motivates calibration. This is an implementation pitfall rather than a library default: of \nnSurveyImpl{} calibration implementations we examined (\Supp{Table~S7}), every one-vs-rest calibrator returns a constant or raises on a single-class split; any wrapper that catches the failure and returns the input produces the takeover. A prior fallback ($\tilde s_{id}=\hat\pi_d$) or exclusion removes the effect; the policy should be an explicit argument of the calibrator.
\end{remark}

\section{Matched Pairs at Scale}
\label{sec:santander}

\begin{figure*}[t]
\centering
\includegraphics[width=.9\textwidth]{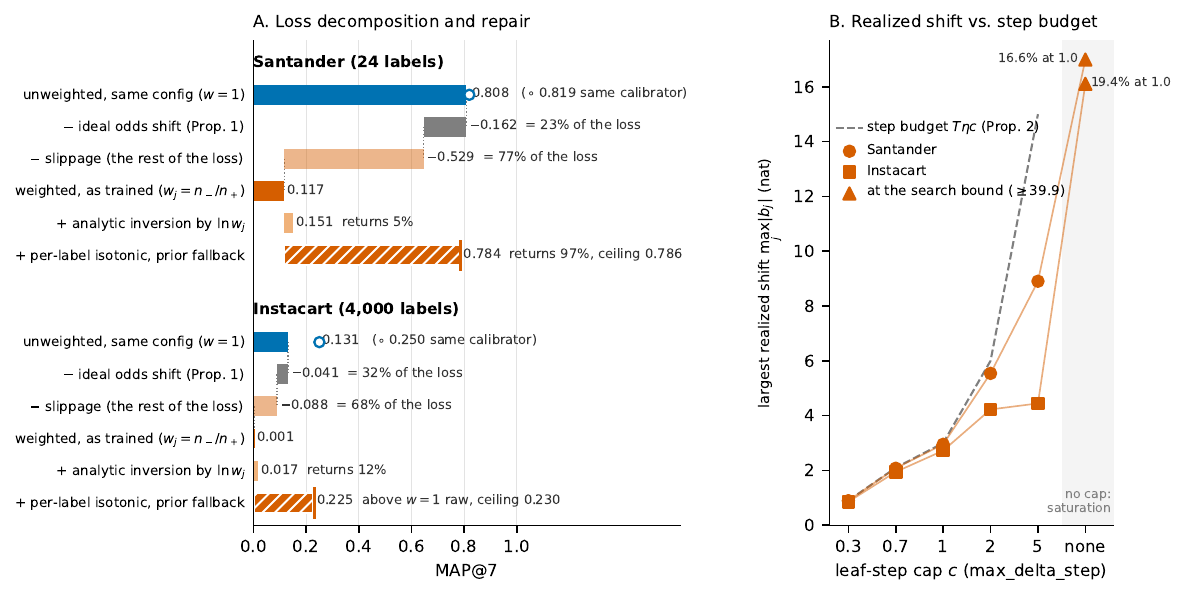}
\Description{Two panels. Panel A: a waterfall of MAP at 7 per dataset from the unweighted model down through the ideal odds shift and the slippage to the trained weighted model, then up by the analytic inversion and by the per-label isotonic repair. Panel B: the largest realized per-label shift against the leaf-step cap, with the step-budget line.}
\caption{\textbf{The textbook odds shift accounts for a minority of a weighted LightGBM ranker's top-$K$ loss on both benchmarks; the rest is slippage, most of which a per-label monotone map recovers and the analytic inversion does not.} \textbf{A.} $\mapseven$ of the unweighted LightGBM ($w{=}1$); the loss to the ideal odds shift, i.e.\ $\ln w_j$ added to its logits (\cref{prop:odds}); the further loss to the weighted model trained with the identical configuration ($w_j{=}n_-/n_+$), which we call slippage; and what its analytic inversion with the known $w_j$ and its per-label isotonic repair with the prior fallback return (tick: in-sample per-label ceiling; open marker: the $w{=}1$ model under the same calibrator). On Santander the shift is \noddsShare\% of the loss $\nNraw-\nSraw$ and the repair returns \nrepairShare\% of it; on Instacart the shift is \nIoddsShare\% of $\nINraw-\nISraw$ and the same calibrator lifts the unweighted model to $\nINplIsoPrior$. One matched pair per dataset; over \nseedN{} train-row draws the weighted Santander bar is $\nseedSraw$ (Robustness), its repaired level stable. \textbf{B.} Largest realized per-label shift $\max_j|b_j|$ (label-free prevalence matching) against the leaf-step cap $c$; the dashed line is the step budget $T\eta c$ of \cref{prop:budget}. Every point lies below the line (\ncapNwithinBound{} of \ncapN{} caps on each dataset) and no cell saturates under any cap; without a cap the shift runs into the search bound (triangle) and \ncapNoneSatOne\% (Santander) and \nIcapNoneSatOne\% (Instacart) of cells sit at exactly $1.0$. A: repair on the 70\% evaluation split, else all test rows; 95\% intervals are narrower than the markers.}
\label{fig:one}
\end{figure*}

\begin{table*}[t]
\caption{Repair ladder on two datasets and two learners ($\mapseven$, test-split protocol, 70\% evaluation split; the ceiling row is fitted and evaluated on all test rows; the deployment protocol, \Supp{Table~S3}, differs by at most $\nprotoMaxDiff$ on the bold rung). Columns: the same configuration trained with the recommended per-label weight ($w{=}r$, $r_j=n_{-,j}/n_{+,j}$) and without it ($w{=}1$), for LightGBM and for a shared-trunk MLP with per-label \texttt{pos\_weight}, on Santander (\nnLabels{} products) and Instacart (top-\nInLabels{} products). ``--'': an unweighted model has no $w_j$ to invert. Bold = per-label isotonic regression with the prior fallback. On all four pairs the weighted model ranks below the popularity baseline and the bold rung returns it to its in-sample ceiling (where that ceiling is below the unweighted raw score, the gap is consistent with the tied cells); the inversion recovers little except on the Instacart MLP, the one pair that realized its shift without saturating (\cref{cor:exact}). (WGBoost and the identity and exclusion fallbacks: \Supp{Table~S1}.)}
\label{tab:ladder}
\footnotesize
\setlength{\tabcolsep}{4pt}
\begin{tabular}{@{}lcccccccc@{}}
\toprule
 & \multicolumn{4}{c}{Santander (\nnLabels{} labels)} & \multicolumn{4}{c}{Instacart (top-\nInLabels{} products)} \\
\cmidrule(lr){2-5}\cmidrule(lr){6-9}
 & \multicolumn{2}{c}{LightGBM} & \multicolumn{2}{c}{MLP} & \multicolumn{2}{c}{LightGBM} & \multicolumn{2}{c}{MLP} \\
 & $w{=}r$ & $w{=}1$ & $w{=}r$ & $w{=}1$ & $w{=}r$ & $w{=}1$ & $w{=}r$ & $w{=}1$ \\
\midrule
raw score & 0.117 & 0.808 & 0.112 & 0.766 & 0.001 & 0.131 & 0.008 & 0.056 \\
analytic inversion with $w_j$ & 0.151 & -- & 0.247 & -- & 0.017 & -- & 0.070 & -- \\
per-label iso, dead$\to$prior & \textbf{0.784} & \textbf{0.819} & \textbf{0.549} & \textbf{0.787} & \textbf{0.225} & \textbf{0.250} & \textbf{0.081} & \textbf{0.056} \\
pooled (shared) isotonic & 0.261 & 0.808 & 0.178 & 0.766 & 0.004 & 0.149 & 0.008 & 0.055 \\
popularity (train prevalence) & 0.535 & 0.535 & 0.535 & 0.535 & 0.056 & 0.056 & 0.056 & 0.056 \\
\midrule
oracle per-label ceiling & 0.786 & 0.822 & 0.548 & 0.791 & 0.230 & 0.269 & 0.085 & 0.056 \\
mean within-label AUC & 0.793 & 0.824 & 0.789 & 0.847 & 0.636 & 0.670 & 0.836 & 0.424 \\
cells at $1.0$ (\%) & 16.6 & 0.0 & 23.3 & 0.0 & 19.4 & 0.0 & 0.0 & 0.0 \\
labels without calibration positives & 3 & 3 & 3 & 3 & 0 & 0 & 0 & 0 \\

\bottomrule
\end{tabular}
\end{table*}

\paragraph{Data and models.} The Santander Product Recommendation data~\cite{KaggleSantander2016} has \nnLabels{} products. We sort the public training file by date and cut it at fixed row-count fractions into \nnTrain{} training, \nnValRows{} validation (held out from fitting, preceding the test period) and \nnTest{} test rows; the label of a row is the set of products added in the following month, and $\mapseven$ follows the Kaggle definition on the rows with at least one added product (\nfracRowsWithPos\% of the \nnEvalRows{} evaluation rows). \Sm and \Nm are LightGBM one-vs-rest models trained by one script with \emph{identical} configuration (60 trees, learning rate $0.05$, 31 leaves, no subsampling, one seed, the same features and split); the only difference is that \Sm sets $w_j=n_{-,j}/n_{+,j}$ per label (between \nwMin{} and \nwMax) and \Nm sets no weight. \Wm (Wasserstein gradient boosting~\cite{Matsubara2024WGBoost}) estimates probabilities directly without weights (\Supp{Table~S1}). The \emph{test-split} (diagnostic) protocol fits calibrators on a random 30\% of the test rows (seed 42, \nnVal{} rows) and evaluates on the remaining 70\%; it is optimistic for deployment, and its gain on the unweighted model ($\nciNprior$) bounds what the protocol itself contributes. \Supp{Table~S5} sweeps the split size. The \emph{deployment} protocol fits calibrators on the validation period and evaluates on all test rows (\Supp{Table~S3}). Under the test-split protocol \nSnDead{} labels have no calibration positives and every calibrated label has at least \ncalPosMinAlive{} (median \ncalPosMedAlive). Bootstrap intervals use $B=\nbootB$ per-row resamples of the \nbootN{} evaluation rows with a positive. One instrument matters for every saturation number below: predictions are stored as single-precision probabilities, so a cell ``at exactly $1.0$'' is one whose raw margin reaches $\nsatMarginF$ nat, the float32 rounding threshold below one. Every saturation share is a property of that stored array, the object a deployed ranker consumes; a pipeline that keeps the raw margin can still order those cells, a calibrator of the probabilities cannot.

The Instacart 2017 data~\cite{Instacart2017} has 49,688 products and, for 131,209 users, an ordered history of prior orders. We use it as a next-order ranking task in the form of the next-basket literature (\cref{sec:related}): one row per user and reference order, features from the orders before it only (8,320 columns), a within-user chronological split (third-last order for training, second-last as validation period, last as test period; 131,209 rows each), and the set of products in the reference order as the label. The label set is the top-\nInLabels{} products by training frequency, the budget of one LightGBM fit per label and configuration; the Kaggle task is not comparable. The weights range from \nIwMin{} to \nIwMax{} ($\ln w_j$ from \nIlnwMin{} to \nIlnwMax); the calibration split has \nInVal{} rows, \nISnDead{} dead labels and a median of \nIcalPosMedAlive{} positives per label, and popularity ranking scores $\nIpopMAP$. The LightGBM pair uses the Santander configuration; both boosters start from the same initial score $\operatorname{logit}\hat\pi_j$, as \cref{prop:budget} assumes. The second learner is a shared-trunk MLP (two hidden layers of width 256, binary cross-entropy with per-label \texttt{pos\_weight}$\in\{1,r_j\}$, 20 epochs, no early stopping, seed 42) on the same rows and labels; on Instacart it scores all \nIMAnLabels{} products in one model, of which the top-\nInLabels{} columns enter \cref{tab:ladder} (all-product ladder: \Supp{S7}).

\paragraph{Decomposition and repair ladder (\cref{fig:one}A, \cref{tab:ladder}).} For this pair the weighted model loses $\nSloss$ $\mapseven$ against \Nm; we express every repair as the share of this loss it returns. Adding $\ln w_j$ to the logits of \Nm is the ideal weighted model of \cref{prop:odds}; it scores $\nwhatifIdeal$, keeps \noverlapWhatif{} of 7 items of the unweighted top-7, and accounts for \noddsShare\% of the observed loss. The trained \Sm keeps only \noverlapS{} of 7 items, has \nSsatOne\% of its cells at exactly $1.0$, and the analytic inversion with the known $w_j$ returns \nelkanShare\% ($\nSelkan$). The ladder measures how much of the remaining distortion a per-label monotone map recovers, rung by rung. An intercept-only logit shift fitted per label on the calibration split returns \noffsetShare\% ($\nSplOffsetPrior$), so for most labels the realized distortion is close to a constant --- but not the constant $\ln w_j$. \Supp{Fig.~S1} shows why label by label: the prevalence-matching shift $b_j$ equals $\ln w_j$ on the labels whose cells never saturate (median $b_j/\ln w_j=\nSbOverLnwMedian$ over \nSnIdentShift{} of \nnLabels{} labels), as \cref{prop:odds} says, and is not identified on the other \nnCappedShift{}, whose cells at exactly $1.0$ no finite shift moves (\Supp{S7}). Per-label isotonic regression with the prior fallback returns \nrepairShare\% ($\nSplIsoPrior$), within $\nSplIsoPriorCeilGap$ of the in-sample per-label ceiling ($\nSceil$). The last \ntieShare\% ($\nSceilGap$) lies above that ceiling and is consistent with the saturated cells: re-ordering only the cells at exactly $1.0$ by the unweighted model's order returns \ntieShareUnw\% of the loss, an oracle order that reads the test labels \ntieShareOracle\% (\Supp{S7}). Three controls guard the comparison. Per-row bootstrap intervals put the gap between \Nm raw and the repaired \Sm at $\nciNminusS$, and under the same calibrator \Nm itself reaches $\nNplIsoPrior$, so the like-for-like gap is $\nNplIsoPrior-\nSplIsoPrior$. The within-label AUC, which no monotone map changes, is $\nNauc$ for \Nm against $\nSauc$ for \Sm. Pooled isotonic leaves \Wm and \Nm unchanged (\cref{rem:shared}) and moves \Sm to $\nSpooled$ only by merging tied cells, which the stable sort then breaks by label index. Trained $\sqrt{r}$ and $10r$ arms behave the same way (raw $\nSsqrtRaw$ and $\nStenrRaw$, repaired $\nSsqrtPlIsoPrior$ and $\nStenrPlIsoPrior$).

\paragraph{Instacart: the same decomposition, a larger role for the calibrator.} The Instacart pair loses $\nISloss$ ($\nINraw$ to $\nISraw$), keeping \nIoverlapS{} of 7 items of the unweighted top-7. The ideal odds shift accounts for \nIoddsShare\% of that loss and already changes the top-7 of \nIfracRowsChangedWhatif\% of rows; the rest is slippage. One thing differs: the same calibrator also lifts the \emph{unweighted} model, from $\nINraw$ to $\nINplIsoPrior$, because after 60 rounds at learning rate $0.05$ the scores of \nInLabels{} labels are not yet comparable across labels and the per-label maps supply that comparability. The cost of the weight is therefore the like-for-like gap after calibration, $\nINplIsoPrior$ against $\nISplIsoPrior$ (ceilings $\nINceil$ against $\nISceil$), not the raw loss. Bootstrap intervals separate every comparison drawn here (repaired weighted model minus its inversion: $\nIciSpriorMinusElkan$). Here the ties dominate: \nIStieCellsExactOne\% of cells sit at exactly $1.0$ and \emph{every} row has at least $K$ of them (mean \nIStieCellsPerRowMean{} of \nInLabels), so the top-$K$ of the raw weighted model is decided entirely by how the sort breaks ties; re-ordering the tied cells by the unweighted order returns \nItieShareUnw\% of the loss (\Supp{S7}). The collapse of the weighted Instacart model is therefore mostly a tie-breaking artefact of saturation (\cref{lem:leaf} gives its single-tree limit), which no separable map of the stored probabilities can undo and a cap prevents at training time.

\begin{figure}[t]
\centering
\includegraphics[width=.8\columnwidth]{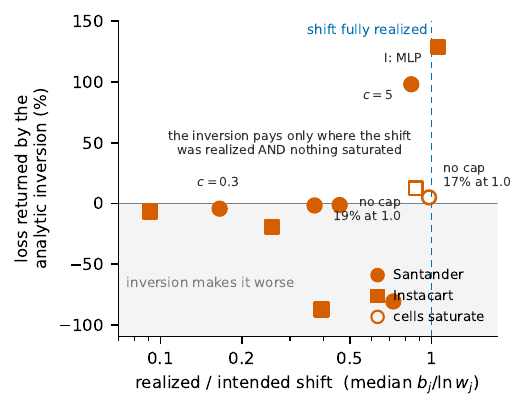}
\Description{Scatter of every trained weighted model: the median realized-over-intended per-label shift on a logarithmic horizontal axis against the share of the top-K loss that the analytic inversion returns. Every model sits at or near a ratio of one or well to its left. The small caps sit well left of one and return nothing or less; the largest Santander cap sits just left of one and returns almost all of the loss, as does the Instacart MLP just right of one; the two uncapped, saturated models sit near one on the horizontal axis but only just above zero.}
\caption{\textbf{Inverting by the known $\ln w_j$ pays only where the learner realized the shift \emph{and} no cell saturated.} One marker per weighted model: both datasets, both learners, and each leaf-step cap. Horizontal: the median of $b_j/\ln w_j$ over the labels with no saturated cell, the only ones on which it is identified. Vertical: what the inversion returns as a share of that model's $\mapseven$ loss against its own unweighted twin, so $100\%$ would restore the unweighted ranking (the losses differ by more than an order of magnitude across models). Filled markers have essentially no cell at exactly $1.0$; open ones carry the share shown. The uncapped Santander booster separates the two conditions: it realized $\nSbOverLnwMedian$ of its intended shift and the inversion still returns only \nelkanShare\%, because \ncapNoneSatOne\% of its cells sit at exactly $1.0$. Omitted: the Santander MLP (shift identifiable on \nMSnIdentShift{} of \nnLabels{} labels) and two Instacart caps with no loss to return. The two conditions are an empirical association, not a proven characterization.}
\label{fig:regime}
\end{figure}

\paragraph{Mechanism: cap sweep and inversion regime (\cref{fig:one}B, \cref{fig:regime}).} Re-training the weighted configuration with \texttt{max\_delta\_step}$=c\in\{0.3,0.7,1,2,5\}$ ($T=\ncapT$, $\eta=\ncapEta$) bounds the realized shift by $T\eta c$, from $\nSmdsPointThreeTEtaC$ to $\nSmdsFiveTEtaC$ nat, against $\ln w_j$ up to $\ncapLnwMax$ on Santander and $\nIcapLnwMax$ on Instacart. \Cref{prop:budget} survives every cell of the sweep: the largest realized shift stays inside the budget at all \ncapN{} caps on both datasets, tightly at the smallest cap ($\nSmdsPointThreebPriorMax$ against $\nSmdsPointThreeTEtaC$ on Santander) and with room to spare at the largest ($\nSmdsFivebPriorMax$ against $\nSmdsFiveTEtaC$), and no cell sits at exactly $1.0$ under any cap. The budget's falsifiable content is the over-correction it forces on the inversion, whose median $\ln w_j-b_j$ falls from $\nSmdsPointThreeoverMedian$ nat at $c=0.3$ to $\nSmdsFiveoverMedian$ nat at $c=5$ on Santander (\Supp{S7}). \Cref{fig:regime} says when the analytic inversion works: the inversion pays only for the two models that realized the shift \emph{and} saturated nothing (Santander $c=5$, raw $\nSmdsFiveraw$ to $\nSmdsFiveelkan$ against $\nNraw$ for \Nm; the Instacart MLP below), it over-corrects where the cap starved the shift and barely moves the uncapped models, whose ties no separable map reaches. The repair does not travel down the sweep unchanged either: per-label isotonic regression with the prior fallback is the best rung at all \nInCapsIsoBest{} Instacart caps and at \nnCapsIsoBest{} of \nnCapsRungs{} Santander caps, but at $c=0.3$ it is the worst rung of that model ($\nSmdsPointThreeplIsoPrior$ against $\nSmdsPointThreeraw$ for doing nothing), because the cap compresses the scores until the per-label map's own in-sample reference ($\nSmdsPointThreeCeil$) sits below raw --- the scarcity that the selection rule of \cref{sec:discussion} exists to catch. At $c=5$ the capped, calibrated weighted arm scores $\nSmdsFiveplIsoPrior$, above the uncapped unweighted model under the same calibrator ($\nNplIsoPrior$), but an unweighted model trained with the same cap scores $\nSmdsFiveNraw$ raw and $\nSmdsFiveNplIsoPrior$ calibrated, level with it, and at the other \ncapNtwinsAhead{} caps the unweighted twin is ahead, by up to $\ncapTwinGapMax$ at $c=0.3$ (\Supp{S7}): the gain is the cap's, not the weight's.

\paragraph{Two learners.} Where it is identified at all, the realized shift of the two learners agrees: the booster's median $b_j/\ln w_j$ is $\nSbOverLnwMedian$ on Santander (\nSnIdentShift{} labels) and $\nISbOverLnwMedian$ on Instacart (\nISnIdentShift), the Instacart MLP's is $\nIMSbOverLnwMedian$ (\nIMSnIdentShift), and a clipping-free estimator (\Supp{S7}) gives $\nSbDirectOverLnwMedian$ and $\nISbDirectOverLnwMedian$ for the two boosters on a larger label set, so an uncapped learner's median lies between half of the intended shift and $\nIMSbOverLnwMedian$ times it. What differs between the pairs is how much of the label set retains a shift (\nMSnIdentShift{} of \nnLabels{} labels on the Santander MLP against \nIMSnIdentShift{} of \nInLabels{} for the Instacart MLP). The \emph{same} ideal shift costs the two Santander pairs very differently ($\nMNraw\to\nMwhatifIdeal$ for the MLP against $\nNraw\to\nwhatifIdeal$ for LightGBM, \nModdsShare\% of the loss against \noddsShare\%): the share measures how large $\ln w_j$ is against the unweighted model's within-row logit gaps. Downstream everything repeats (\nMSsatOne\% of cells at exactly $1.0$, the inversion returns \nMelkanShare\%, the bold rung ($\nMSplIsoPrior$) lands on its ceiling ($\nMSceil$), \nMtieShare\% of the loss below the unweighted MLP's raw score). Instacart adds the case the recipe has to cover: the \emph{unweighted} MLP never learns a within-row signal (AUC $\nIMNauc$, score $\nIMNraw$, the popularity baseline to three decimals), the weighted MLP does (AUC $\nIMSauc$) and pays for it in the cross-label scale ($\nIMSraw$ raw). This is the one pair that meets both conditions of \cref{fig:regime} --- realized shift at the intended level, saturation negligible --- and the inversion accordingly works, taking it from $\nIMSraw$ to $\nIMSelkan$, past the unweighted twin (the all-product MLP meets the same conditions and inverts too, \Supp{S7}); per-label isotonic regression does better still ($\nIMSplIsoPrior$). Here the weight made the optimization move and was ruinous for the ranking until the marginals were restored.

\paragraph{The fallback decides the sign.} With the identity fallback, per-label isotonic regression yields $\nSplIsoId$ on \Sm (below the popularity baseline $\npopMAP$) and lowers \Wm from $\nWraw$ to $\nWplIsoId$: the \nSnDead{} dead labels keep their raw scores (mean $\nWdeadRawMean$ on \Wm) while every calibrated label shrinks to its prevalence (largest mean $\nWaliveCalMax$), so all three sit in the top 7 of \nWdeadAllRowsPct\% of rows, the takeover of \cref{prop:dead}. Since scarce calibration positives hurt even with the prior fallback (\cref{sec:dose}), the two effects must be kept apart. With the prior fallback the same calibrator gives $\nSplIsoPrior$ on \Sm, $\nWplIsoPrior$ on \Wm ($\nciWprior$) and $\nNplIsoPrior$ on \Nm ($\nciNprior$), each within $0.01$ of its ceiling; exclusion behaves like the prior. On the all-product Instacart MLP the policy decides everything: \nIMASnDead{} of \nIMAnLabels{} products have no calibration positive, the identity fallback returns the weighted model exactly to its raw score ($\nIMASplIsoId$) and the prior fallback repairs it to $\nIMASplIsoPrior$ against a popularity baseline of $\nIMApopMAP$ (\Supp{S7}).

\paragraph{Robustness.} Under the deployment protocol \Sm goes from $\ndepSraw$ to $\ndepSplIsoPrior$; the two protocols differ on the bold rung by at most $\nprotoMaxDiff$ (Santander, \Supp{Table~S3}) and $\nIprotoMaxDiffPair$ (Instacart, \Supp{S7}), with the order of the ladder unchanged. Retraining both LightGBM pairs on \nseedN{} independent 90\% draws of the training rows moves only the depth of the weighted Santander collapse ($\nseedSraw$ raw) and the two rungs that read that depth; the repaired score ($\nseedSplIsoPrior$) and the odds share ($\nseedOddsShare\%$) are stable, and Instacart is stable throughout (\Supp{S7}).
\section{Dose-Response on Public Benchmarks}
\label{sec:dose}

\begin{figure*}[t]
\centering
\includegraphics[width=.86\textwidth]{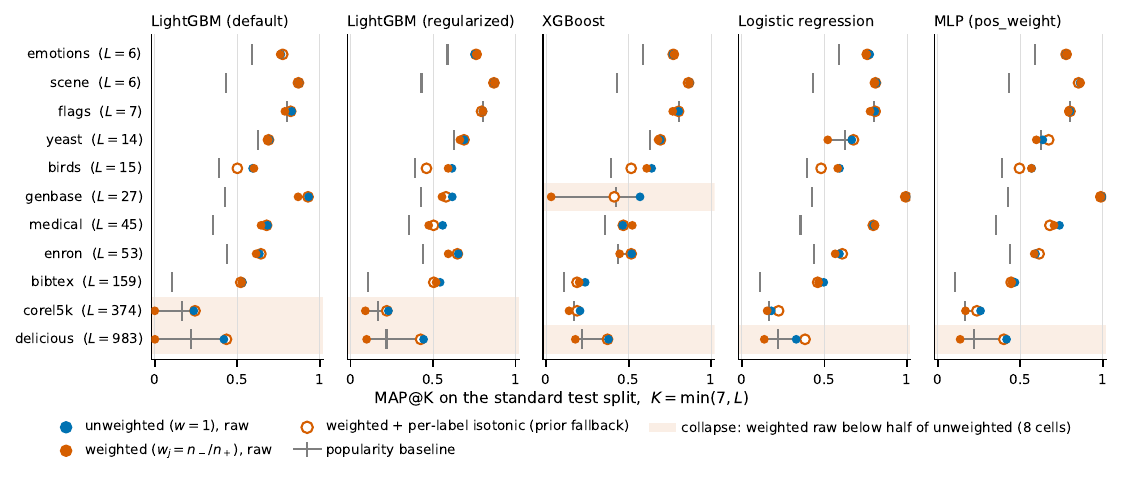}
\Description{Five panels, one per learner, each with eleven rows of MULAN datasets sorted by label count. In each row a blue dot marks the unweighted model, an orange dot the weighted model and an open orange circle the weighted model after per-label isotonic calibration; a grey tick marks the popularity baseline. Rows where the weighted dot falls below half of the blue dot are shaded: all five learners on delicious, two on Corel5k, XGBoost on genbase.}
\caption{\textbf{On public data the recommended weight is benign or mildly harmful on most datasets and collapses the ranking where the labels are many; on delicious and Corel5k per-label calibration with the prior fallback returns the collapsed models to the unweighted level.} $\mapk$ on the standard test split of \nnMulan{} MULAN datasets~\cite{Tsoumakas2011MULAN} (rows, sorted by label count $L$; $K=\min(7,L)$) for \nnLearners{} one-vs-rest learners (panels), each trained without weights ($w{=}1$, filled blue) and with the recommended weight ($w_j=n_{-,j}/n_{+,j}$, filled orange); the open circle is the weighted model after per-label isotonic regression with the prior fallback, and the grey tick the popularity baseline. Shaded rows are the \nnCollapseCells{} of \nnDoseCells{} cells in which the weighted model keeps less than half of its unweighted $\mapk$ (a threshold chosen for the figure; at $0.3$, $0.7$ and $0.9$ the count is \nnCollapseCellsAtPointThree, \nnCollapseCellsAtPointSeven{} and \nnCollapseCellsAtPointNine): all \nnDelCollapseLearners{} learners on delicious ($L=\ndelL$), \nnCorelCollapseLearners{} on Corel5k ($L=\ncorelL$), \nnOtherCollapseCells{} on genbase. Calibration split: 30\% of the training file (seed 42); \Supp{S6} has every cell and four weight magnitudes.}
\label{fig:mulan}
\end{figure*}

\begin{table*}[t]
\caption{\textbf{At moderate imbalance the sign of the weight's effect is dataset-dependent; where the labels are many (bottom two rows) the weighted model collapses and per-label calibration returns it to the unweighted level.} Default LightGBM one-vs-rest on \nnMulan{} MULAN datasets, $\mapk$ with $K=\min(7,L)$ on the standard test split. $n_{\mathrm{fit}}$: training rows after removing the calibration split (30\% of the training file, at least 50 rows, seed 42); ``cal.\ pos.'': median calibration positives per label; ``pop.'': ranking by training prevalence. Unweighted model ($w{=}1$): raw and per-label isotonic with prior fallback. Weighted model ($w_j=n_-/n_+$): raw, what-if prediction (unweighted logits $+\ln w_j$), analytic inversion, per-label isotonic with prior fallback, in-sample ceiling. Last columns: mean number of distinct scores per label. $\min r<1$ marks a majority-positive label; a label with no positive in the fit split is recorded as $r=1$ (genbase, medical, Corel5k).}
\label{tab:mulan}
\footnotesize
\setlength{\tabcolsep}{3pt}
\begin{tabular}{@{}lrrrrrrrr|cc|ccccc|cc@{}}
\toprule
 & & & & & & & & & \multicolumn{2}{c|}{$w{=}1$} & \multicolumn{5}{c|}{$w{=}r$} & \multicolumn{2}{c}{distinct} \\
dataset & $L$ & $K$ & $n_{\mathrm{fit}}$ & $n_{\mathrm{test}}$ & cal.\ pos. & $\min r$ & $\max r$ & pop. & raw & pl-iso & raw & what-if & inversion & pl-iso & ceiling & $w{=}1$ & $w{=}r$ \\
\midrule
emotions & 6 & 6 & 274 & 202 & 32 & 1.5 & 3 & 0.588 & 0.764 & 0.776 & 0.760 & 0.757 & 0.763 & 0.775 & 0.805 & 202 & 202 \\
scene & 6 & 6 & 848 & 1196 & 60 & 3.5 & 6 & 0.430 & 0.874 & 0.873 & 0.873 & 0.877 & 0.871 & 0.870 & 0.883 & 1195 & 1195 \\
flags & 7 & 7 & 79 & 65 & 22 & 0.3 & 6 & 0.802 & 0.829 & 0.805 & 0.788 & 0.762 & 0.838 & 0.823 & 0.870 & 58 & 59 \\
yeast & 14 & 7 & 1050 & 917 & 126 & 0.3 & 80 & 0.626 & 0.694 & 0.695 & 0.688 & 0.683 & 0.691 & 0.689 & 0.704 & 917 & 917 \\
birds & 15 & 7 & 226 & 323 & 5 & 7.7 & 112 & 0.392 & 0.593 & 0.549 & 0.601 & 0.593 & 0.593 & 0.500 & 0.675 & 322 & 321 \\
genbase & 27 & 7 & 325 & 199 & 5 & 1.0 & 324 & 0.425 & 0.932 & 0.946 & 0.868 & 0.927 & 0.876 & 0.928 & 0.952 & 69 & 51 \\
medical & 45 & 7 & 234 & 645 & 1 & 1.0 & 233 & 0.355 & 0.683 & 0.651 & 0.644 & 0.653 & 0.674 & 0.677 & 0.764 & 261 & 205 \\
enron & 53 & 7 & 787 & 579 & 6 & 0.9 & 786 & 0.437 & 0.634 & 0.642 & 0.614 & 0.596 & 0.636 & 0.643 & 0.680 & 510 & 496 \\
bibtex & 159 & 7 & 3416 & 2515 & 17 & 5.8 & 227 & 0.107 & 0.531 & 0.530 & 0.527 & 0.537 & 0.521 & 0.521 & 0.580 & 2327 & 2349 \\
corel5k & 374 & 7 & 3150 & 500 & 4 & 1.0 & 3149 & 0.166 & 0.236 & 0.241 & 0.000 & 0.164 & 0.001 & 0.244 & 0.314 & 276 & 244 \\
delicious & 983 & 7 & 9044 & 3185 & 26 & 1.5 & 903 & 0.218 & 0.419 & 0.434 & 0.001 & 0.237 & 0.003 & 0.434 & 0.482 & 2576 & 2135 \\

\bottomrule
\end{tabular}
\end{table*}

\paragraph{MULAN benchmarks (\cref{fig:mulan}, \cref{tab:mulan}).} We retrain one-vs-rest models on \nnMulan{} MULAN datasets~\cite{Tsoumakas2011MULAN} (every standard-split dataset in the distribution except the two largest, tmc2007 and Corel16k) with $w_j\in\{1,\sqrt{r_j},r_j,10r_j\}$ and \nnLearners{} learners: LightGBM in a default and a strongly regularized configuration, XGBoost, $\ell_2$-regularized logistic regression with sample weights, and a two-layer MLP with per-label \texttt{pos\_weight}. The largest weight the sweep reaches is $\nnMulanMaxrMax$ at $w{=}r$, three orders of magnitude below Santander's; it is a dose-response over imbalance magnitude that also covers the label counts of \cref{sec:santander}. Four regularities hold.

(i) \emph{The collapse reproduces on public data.} At $w{=}r$ the raw score falls below half its unweighted value in \nnCollapseCells{} of \nnDoseCells{} learner cells, on \ncollapseDatasets{} alone. With the default LightGBM delicious goes from $\ndelRawWone$ to $\ndelRawWr$ and Corel5k from $\ncorelRawWone$ to $\ncorelRawWr$, both below their popularity baselines --- the shape of \cref{tab:ladder} on a public benchmark. The loss is between labels rather than within them (within-label AUC moves only from $\ndelAucWone$ to $\ndelAucWr$ on delicious), and the repair follows the same pattern: per-label isotonic regression with the prior fallback reaches $\ndelPlIsoPriorWr$ and $\ncorelPlIsoPriorWr$, at or above what the \emph{unweighted} model attains under the same calibrator, while a shared map reaches only $\ndelPooledWr$ and $\ncorelPooledWr$ and the inversion only $\ndelInvWr$ and $\ncorelInvWr$. Over all \nnRepairCollapseCells{} learner cells of delicious and Corel5k the repaired weighted model is within $\nmaxRepairGapCollapse$ of its unweighted arm and the shared map is never within $\nminPooledGapCollapse$ of it. Corel5k also carries the dead-label half: \ncorelDeadCalDose{} of its \ncorelL{} labels have no calibration positives, and the identity fallback gives $\ncorelPlIsoIdWr$ against the prior's $\ncorelPlIsoPriorWr$ on the default LightGBM --- \cref{prop:dead} on a public benchmark, though the gap runs down to $\ncorelIdGapMin$ on other learners. We have no threshold to offer: enron has the weight spread ($\max r=\nenronMaxr$) and bibtex the label count ($L=\nbibtexL$) without collapsing.

(ii) \emph{Elsewhere the sign is learner-dependent.} On the \nnMulanModest{} datasets with $\max r<100$ the default LightGBM change at $w{=}r$ is a fall of between $\nmodestDefaultBest$ and $\nmodestDefaultWorst$, while across the \nnLearners{} learners the sign varies and the spread widens to $\nmodestAnyWorst$ and $\nmodestAnyBest$. The default-learner reading is in line with the report that class weighting leaves calibration nearly intact at imbalance ratios below 70~\cite{Liu2026Resampling}.

(iii) \emph{The what-if prediction is close where nothing saturates.} It is within $0.03$ of the trained weighted model on \nnWhatifClose{} of \nnMulan{} datasets; it under-predicts the loss where the trained model saturates (genbase: $\ngenbaseDefaultWhatifWr$ predicted, $\ngenbaseDefaultRawWr$ trained) and over-predicts it where regularization shrinks the realized shift (enron, strongly regularized), where the inversion then recovers the unweighted ranking ($\nenronStrongInvWr$ against $\nenronStrongRawWone$ at $w{=}1$) --- the small-data counterpart of \cref{fig:regime}.

(iv) \emph{Ties can arrive without saturation.} On medical, strongly regularized, the distinct scores per label fall from $\nmedStrongDistinctWone$ to $\nmedStrongDistinctWtenr$ between $w{=}1$ and $10r$ with $\nmedStrongSatWtenr$\% of cells saturated: on data this small the ties that defeat a separable map come from coarse leaves. Leaf-saturation checks and a synthetic study with known marginals (\Supp{Tables~S8--S10}) agree: a realizable learner inverts exactly, LightGBM fails where its cells saturate.

\paragraph{The negative result and the selection rule.} With few calibration positives, per-label isotonic regression with the prior fallback \emph{hurts} an unweighted model. With the default LightGBM at $w{=}1$ it lowers $\mapk$ by more than $0.005$ on \nnHurtDefault{} of \nnMulan{} datasets, and on \nnHurtTauZeroCI{} of them the paired bootstrap interval excludes zero (birds $\nbirdsRawWone\to\nbirdsPlIsoPriorWone$ with a median of \ncalPosMedBirds{} calibration positives per label; medical $\nmedRawWone\to\nmedPlIsoPriorWone$ with \ncalPosMedMedical); the median number of calibration positives does not separate the cases (genbase and enron gain with \ncalPosMedGenbase{} and \ncalPosMedEnron), so no fixed threshold can be recommended. The recipe's answer is a rule fixed on the calibration split alone: \ntauFolds-fold cross-validation chooses among per-label isotonic regression with positives threshold $\tau\in\{\ntauList\}$ (labels with at most $\tau$ calibration positives keep the prior), the shared map, and no calibration (\Supp{Table~S6}). For the default LightGBM at $w{=}1$ the rule keeps the raw scores on birds, medical and flags, keeps the per-label gains on genbase and enron, loses to raw with an interval excluding zero on \nnHurtAfterTau{} of \nnMulan{} datasets, and gains with one on delicious ($\ndelCiTauSel$). It is not infallible: across all \nnTauCells{} cells the selection ends more than $0.01$ below raw in \nnTauSelBelowRaw{} (against \nnTauZeroBelowRaw{} for the fixed $\tau=0$ rule); in its worst cell ($\ntauWorstSelMinusRaw$) raw and the selected candidate are indistinguishable out of fold yet differ on the test file (\Supp{S2}).

\section{Related Work}
\label{sec:related}

\paragraph{The odds shift, and what a learner does with it.} Elkan's identity~\cite{Elkan2001} and the prior-shift correction of~\cite{Saerens2002} describe a positive-class weight as a shift of the log-odds by $\ln w$; logit adjustment applies that shift post hoc in long-tailed classification~\cite{Menon2021LogitAdjust}, label-shift adaptation estimates the shift from a black-box predictor~\cite{Lipton2018LabelShift} or combines it with calibration~\cite{Alexandari2020MLLS}, of which our prevalence-matching shift is the per-label special case, and the inversion the identity suggests is the first repair a practitioner reaches for; Caplin et al.~\cite{Caplin2022ClassWeights} already observe that it fails on a trained binary classifier and recover probabilities through a model of the learner, whereas we measure the realized shift directly and bound it under a cap. The identity describes the population minimizer. Importance weights lose their effect on learners that can fit the weighted objective~\cite{Byrd2019}; cost-sensitive boosting recovers the posterior ratio only near the decision boundary~\cite{MasnadiShirazi2011CSBoost}, is matched by calibrating an unweighted booster and thresholding~\cite{Nikolaou2016CSBoost}, and depends on the model being well specified~\cite{Dmochowski2010CostSensitive}, all for one binary decision rather than a ranking across labels; probability estimates are unreliable under imbalance~\cite{Wallace2012} and under its corrections~\cite{VandenGoorbergh2022}; a binary tree-ensemble study at imbalance ratios up to 70 finds, in a class-weight control, that class weighting, unlike resampling, barely moves calibration~\cite{Liu2026Resampling}; it measures calibration rather than a cross-label ranking, and our default-learner results at $\max r<100$ are consistent with it while other learners are not. Long-tailed recognition adjusts the classifier after unweighted training~\cite{Kang2020Decoupling,Menon2021LogitAdjust} within one softmax; we rank labels across independent one-vs-rest models, where the per-label weight is the distortion. Elkan already notes that rebalancing barely changes a decision tree, an early qualitative form of slippage; none of these measures how far a gradient booster~\cite{Friedman2001GBM} realizes the shift when $w_j$ is in the thousands.

\paragraph{Consistency and top-$K$.} One-vs-rest surrogates are consistent for the rank loss~\cite{Dembczynski2012Univariate} and, with a proper loss, for precision@$k$~\cite{Menon2019Reductions,Wydmuch2018PLT}; propensity-scored metrics weight labels deliberately~\cite{Jain2016XMLLoss}, which \cref{cor:odds}(c) shows training weights do implicitly.

\paragraph{Calibration versus ranking.} Post-hoc calibration~\cite{Platt1999,Zadrozny2002,NiculescuMizil2005,Kull2017Beta,Guo2017Calibration} optimizes a probability metric, and its effect on a ranking depends on whether the map is shared or class-wise: shared class-wise binning recovers the top-5 accuracy that independent per-class maps lose on ImageNet~\cite{Patel2021IMax}, intra-order-preserving maps are rank-preserving by construction~\cite{Rahimi2020IntraOrder}, and in extreme multi-label classification a shared monotone map keeps precision@$k$ untouched while per-label calibration is deferred to future work~\cite{Ullah2025XMLCCalib}; label-wise monotone reranking after calibration~\cite{Tae2026MRP} targets the residual reliability of fixed decisions, not the cross-label scale that a training weight distorts, and has no policy for a label without calibration positives. Per-class residual offsets are Bayes-optimal for reranking when the distortion is class-separable~\cite{Wang2026BeyondLA}, the case our intercept-only rung tests, but that residual is defined on a base model's shortlist and is not traced to a training weight, whereas slippage is the gap between an intended $\ln w_j$ and what the learner realized, bounded by the step budget and broken by saturation; recommender systems calibrate within rank groups~\cite{Sato2024TopN}; automated calibrator selection~\cite{Abdelrahman2025SmartCal} chooses by a probability metric rather than by the deployed ranking metric.

\paragraph{Next-basket recommendation.} Next-basket recommendation scores every item from the ordered history of prior baskets and cuts the list at a fixed $K$~\cite{Yu2016DREAM}; its benchmark canon~\cite{Li2023NBRRealityCheck} and the models it compares~\cite{Hu2020TIFUKNN,Ariannezhad2022ReCANet} use the Instacart 2017 data~\cite{Instacart2017} in this form, with the last basket held out and frequency rankers as the reference. Our Instacart experiments follow that formulation but ask what per-label weights do to a fixed one-vs-rest learner, not which model predicts the next basket best; the popularity row of \cref{tab:ladder} is that literature's frequency baseline.

\section{Discussion}
\label{sec:discussion}

\paragraph{Recipe.} Five steps, in order, for a one-vs-rest ranker over many rare labels.
\begin{enumerate}[leftmargin=*,nosep]
\item \textbf{Estimate marginals with an unweighted proper loss.} If imbalance makes optimization hard, prefer a weight common to all labels, which is rank-neutral by \cref{cor:odds}(a).
\item \textbf{If label-specific weights are kept, cap the leaf step.} A cap removes the saturation ties, which nothing post hoc can repair, at the price of leaving the intended shift unrealized: the inversion is then wrong in the other direction and calibration is what fixes the ranking (\cref{fig:regime}). A small cap compresses the scores until no per-label map recovers the cross-label order ($c=0.3$ above), which is why step 5 chooses the calibrator rather than fixing it.
\item \textbf{Calibrate on a held-out split from the training period}, never from the evaluation period.
\item \textbf{Send labels without calibration positives to their prior, not through the identity.} This choice, not the threshold or the split, decides whether the repair helps or hurts; where positives are merely scarce no $\tau$ makes a per-label map safe (birds, medical).
\item \textbf{Choose $\tau$, per-label versus shared, and whether to calibrate at all, by cross-validation inside the calibration split} (\Supp{Table~S6}): the objective is the deployed metric and the rule never touches test labels. Where that split does not predict the test data, no calibration is the safe default for an unweighted model; the shared map is not (more than $0.01$ below raw in \nnSharedBelowRaw{} of \nnTauCells{} cells).
\end{enumerate}
Where the pipeline allows it, keep the raw margin rather than the stored probability: it preserves the order inside the saturated cells, though not the cross-label scale. Without labels, the share of cells at exactly $1.0$ and the number of distinct scores per label flag saturation and predict that the inversion will fail.

\paragraph{Limitations.}
\begin{itemize}[leftmargin=*,nosep]
\item \emph{Scope of the large-scale evidence.} Two recommendation benchmarks with \nnLabels{} and \nInLabels{} labels, two learner families and one headline weight regime ($w_j=n_-/n_+$ up to \nwMax); the $\sqrt{r}$ and $10r$ arms and the capped unweighted twins exist on Santander only, and $\mapk$ is the only metric reported (recall and NDCG at $K$, the next-basket canon, are not). MULAN replicates the collapse but not its scale (largest $r=\nnMulanMaxrMax$; Corel5k and delicious have \ncorelNtestDose{} and \ndelNtestDose{} test rows) and does not locate the boundary; the two largest MULAN datasets were not run.
\item \emph{Only the step budget is tight.} Without a cap we bound nothing; the realized shift is measurable only where no cell saturates, and the two estimators of the shift disagree on how much of it is realized ($\nSbDirectOverLnwMedian$ against $\nSbOverLnwMedian$ of $\ln w_j$ on Santander, $\nISbDirectOverLnwMedian$ against $\nISbOverLnwMedian$ on Instacart) because they read different label sets. The two conditions of \cref{fig:regime} are an empirical association, not a proven characterization.
\item \emph{Single runs and a fixed configuration.} Each matched pair is a single run. Retraining the Santander pair on a second machine with the same data, script, configuration and LightGBM version did not reproduce it bit for bit: the weighted model scored $\nSSretrainRaw$ raw against $\nSraw$, the unweighted one $\nSNretrainRaw$ against $\nNraw$ (\Supp{S7}). Together with the \nseedN{} draws ($\nseedSraw$) this says that the depth of the weighted Santander collapse is run-dependent while every other rung is stable; the canonical arrays are archived with the release. The MLP pairs are single stochastic runs at one width and one epoch budget, evaluated under the test-split protocol only; the LightGBM configuration was fixed on Santander and not tuned for Instacart.
\item \emph{What the ceiling is, and where the selection rule fails.} The oracle ceiling is the best per-label isotonic fit in squared loss read off at $\mapk$, not a bound on all post-hoc procedures; row-dependent calibrators~\cite{Rahimi2020IntraOrder} are outside the separable class. No threshold in calibration positives separates harmful from harmless per-label calibration, and the cross-validated choice ends more than $0.01$ below raw in \nnTauSelBelowRaw{} of \nnTauCells{} cells, where the calibration split does not predict the test file. Its grid ends at $\tau=\ntauMax$, which the rule picks in \nnTauChoiceMaxTau{} of its \nnTauChoicePerLabel{} per-label choices; a wider grid was not run. The prior fallback is a floor for stationary catalogues: a label with test positives but no calibration positives is sent to the bottom of every row and the rule cannot see the loss; on Santander all \nnDeadCalAlsoTest{} dead labels are also dead in the test split.
\end{itemize}

\paragraph{Conclusion.} Per-label class weights change the target of one-vs-rest training from the marginal to a weighted odds, and on two recommendation benchmarks the textbook odds shift is the minor part of the resulting top-$K$ loss for the boosted pairs (\noddsShare\% on Santander, \nIoddsShare\% on Instacart); the major part is odds-shift slippage, a realized shift that differs from $\ln w_j$ plus ties from saturation, which the analytic inversion cannot undo. What varies across learners and datasets is how much of the label set saturates, so the repair is what carries across learners and the inversion is not. Per-label calibration with a prior fallback supplies it up to the residue the ties impose, and the fallback, not the calibrator, decides its sign. Where the weight is not needed for optimization, a label-specific weight bought nothing for the ranking on our benchmarks; where it is, it should be followed by per-label calibration before the top-$K$ decision.

\section*{Reproducibility}
All scores, calibrators, tables and figures are produced by the released experiment and publication sources (\url{https://github.com/souldrive7/odds-shift-slippage}) from fixed prediction arrays (SHA-256 recorded in the result files; regenerated from the public Kaggle data by the released training scripts) and the public MULAN ARFF files (content hashes recorded). Every run's configuration and seed are recorded, the dead-label policy is a named argument of every calibrator, and the supplement is an ancillary file of this preprint.

\section*{Acknowledgments}
The author used Claude (Anthropic) and Codex (OpenAI) as coding, language-editing and manuscript-review tools; all analyses, claims, numbers and citations were produced or verified by the author, who is solely responsible for the content.

\bibliographystyle{ACM-Reference-Format}
\bibliography{refs}


\begin{thebibliography}{44}


\ifx \showCODEN    \undefined \def \showCODEN     #1{\unskip}     \fi
\ifx \showISBNx    \undefined \def \showISBNx     #1{\unskip}     \fi
\ifx \showISBNxiii \undefined \def \showISBNxiii  #1{\unskip}     \fi
\ifx \showISSN     \undefined \def \showISSN      #1{\unskip}     \fi
\ifx \showLCCN     \undefined \def \showLCCN      #1{\unskip}     \fi
\ifx \shownote     \undefined \def \shownote      #1{#1}          \fi
\ifx \showarticletitle \undefined \def \showarticletitle #1{#1}   \fi
\ifx \showURL      \undefined \def \showURL       {\relax}        \fi
\providecommand\bibfield[2]{#2}
\providecommand\bibinfo[2]{#2}
\providecommand\natexlab[1]{#1}
\providecommand\showeprint[2][]{arXiv:#2}
\makeatletter
\@ifundefined{NAT@parse@date}{}{\let\NAT@parse@date@orig\NAT@parse@date}
\@ifundefined{NAT@parse@date}{}{\def\NAT@parse@date#1#2#3#4#5#6@@{\NAT@parse@date@orig#1#2#3#4#5#6@@\def\NAT@tempyear{0000}\def\NAT@tempexlab{{?}}\ifx\NAT@year\NAT@tempyear\ifx\NAT@exlab\NAT@tempexlab\def\NAT@date{[n.\,d.]}\else\edef\NAT@date{[n.\,d.]\NAT@exlab}\fi\fi}}
\makeatother

\bibitem[Abdelrahman et~al\mbox{.}(2025)]%
        {Abdelrahman2025SmartCal}
\bibfield{author}{\bibinfo{person}{Mohamed~Maher Abdelrahman},
  \bibinfo{person}{Osama~Fayez Oun}, \bibinfo{person}{Youssef Medhat},
  \bibinfo{person}{Mariam~Magdy Elseedawy}, \bibinfo{person}{Yara~Mostafa
  Marei}, \bibinfo{person}{Abdullah Ibrahim}, {and}
  \bibinfo{person}{Radwa~Mohamed El~Shawi}.} \bibinfo{year}{2025}\natexlab{}.
\newblock \showarticletitle{SmartCal: A Novel Automated Approach to Classifier
  Probability Calibration}. In \bibinfo{booktitle}{\emph{Proceedings of the
  Fourth International Conference on Automated Machine Learning (AutoML)}}
  \emph{(\bibinfo{series}{Proceedings of Machine Learning Research},
  Vol.~\bibinfo{volume}{293})}. \bibinfo{pages}{1--14}.
\newblock


\bibitem[Alexandari et~al\mbox{.}(2020)]%
        {Alexandari2020MLLS}
\bibfield{author}{\bibinfo{person}{Amr Alexandari}, \bibinfo{person}{Anshul
  Kundaje}, {and} \bibinfo{person}{Avanti Shrikumar}.}
  \bibinfo{year}{2020}\natexlab{}.
\newblock \showarticletitle{Maximum Likelihood with Bias-Corrected Calibration
  is Hard-To-Beat at Label Shift Adaptation}. In
  \bibinfo{booktitle}{\emph{Proceedings of the 37th International Conference on
  Machine Learning (ICML)}} \emph{(\bibinfo{series}{Proceedings of Machine
  Learning Research}, Vol.~\bibinfo{volume}{119})}. \bibinfo{publisher}{PMLR},
  \bibinfo{pages}{222--232}.
\newblock
\urldef\tempurl%
\url{https://proceedings.mlr.press/v119/alexandari20a.html}
\showURL{%
\tempurl}


\bibitem[Ariannezhad et~al\mbox{.}(2022)]%
        {Ariannezhad2022ReCANet}
\bibfield{author}{\bibinfo{person}{Mozhdeh Ariannezhad}, \bibinfo{person}{Sami
  Jullien}, \bibinfo{person}{Ming Li}, \bibinfo{person}{Min Fang},
  \bibinfo{person}{Sebastian Schelter}, {and} \bibinfo{person}{Maarten de
  Rijke}.} \bibinfo{year}{2022}\natexlab{}.
\newblock \showarticletitle{{ReCANet}: A Repeat Consumption-Aware Neural
  Network for Next Basket Recommendation in Grocery Shopping}. In
  \bibinfo{booktitle}{\emph{Proceedings of the 45th International ACM SIGIR
  Conference on Research and Development in Information Retrieval (SIGIR
  '22)}}. \bibinfo{publisher}{ACM}, \bibinfo{address}{Madrid, Spain},
  \bibinfo{pages}{1240--1250}.
\newblock
\href{https://doi.org/10.1145/3477495.3531708}{doi:\nolinkurl{10.1145/3477495.3531708}}


\bibitem[Byrd and Lipton(2019)]%
        {Byrd2019}
\bibfield{author}{\bibinfo{person}{Jonathon Byrd} {and}
  \bibinfo{person}{Zachary~C. Lipton}.} \bibinfo{year}{2019}\natexlab{}.
\newblock \showarticletitle{What is the Effect of Importance Weighting in Deep
  Learning?}. In \bibinfo{booktitle}{\emph{Proceedings of the 36th
  International Conference on Machine Learning (ICML)}}
  \emph{(\bibinfo{series}{Proceedings of Machine Learning Research},
  Vol.~\bibinfo{volume}{97})}. \bibinfo{pages}{872--881}.
\newblock


\bibitem[Caplin et~al\mbox{.}(2022)]%
        {Caplin2022ClassWeights}
\bibfield{author}{\bibinfo{person}{Andrew Caplin}, \bibinfo{person}{Daniel
  Martin}, {and} \bibinfo{person}{Philip Marx}.}
  \bibinfo{year}{2022}\natexlab{}.
\newblock \bibinfo{title}{Calibrating for Class Weights by Modeling Machine
  Learning}.
\newblock \bibinfo{howpublished}{arXiv:2205.04613}.
\newblock


\bibitem[Chen and Guestrin(2016)]%
        {Chen2016XGBoost}
\bibfield{author}{\bibinfo{person}{Tianqi Chen} {and} \bibinfo{person}{Carlos
  Guestrin}.} \bibinfo{year}{2016}\natexlab{}.
\newblock \showarticletitle{XGBoost: A Scalable Tree Boosting System}. In
  \bibinfo{booktitle}{\emph{Proceedings of the 22nd ACM SIGKDD International
  Conference on Knowledge Discovery and Data Mining (KDD)}}.
  \bibinfo{pages}{785--794}.
\newblock
\href{https://doi.org/10.1145/2939672.2939785}{doi:\nolinkurl{10.1145/2939672.2939785}}


\bibitem[Dembczy{\'n}ski et~al\mbox{.}(2012)]%
        {Dembczynski2012Univariate}
\bibfield{author}{\bibinfo{person}{Krzysztof Dembczy{\'n}ski},
  \bibinfo{person}{Wojciech Kot{\l}owski}, {and} \bibinfo{person}{Eyke
  H{\"u}llermeier}.} \bibinfo{year}{2012}\natexlab{}.
\newblock \showarticletitle{Consistent Multilabel Ranking through Univariate
  Loss Minimization}. In \bibinfo{booktitle}{\emph{Proceedings of the 29th
  International Conference on Machine Learning (ICML)}}.
\newblock


\bibitem[Dmochowski et~al\mbox{.}(2010)]%
        {Dmochowski2010CostSensitive}
\bibfield{author}{\bibinfo{person}{Jacek~P. Dmochowski}, \bibinfo{person}{Paul
  Sajda}, {and} \bibinfo{person}{Lucas~C. Parra}.}
  \bibinfo{year}{2010}\natexlab{}.
\newblock \showarticletitle{Maximum Likelihood in Cost-Sensitive Learning:
  Model Specification, Approximations, and Upper Bounds}.
\newblock \bibinfo{journal}{\emph{Journal of Machine Learning Research}}
  \bibinfo{volume}{11} (\bibinfo{year}{2010}), \bibinfo{pages}{3313--3332}.
\newblock


\bibitem[Elkan(2001)]%
        {Elkan2001}
\bibfield{author}{\bibinfo{person}{Charles Elkan}.}
  \bibinfo{year}{2001}\natexlab{}.
\newblock \showarticletitle{The Foundations of Cost-Sensitive Learning}. In
  \bibinfo{booktitle}{\emph{Proceedings of the 17th International Joint
  Conference on Artificial Intelligence (IJCAI)}}. \bibinfo{pages}{973--978}.
\newblock


\bibitem[Friedman(2001)]%
        {Friedman2001GBM}
\bibfield{author}{\bibinfo{person}{Jerome~H. Friedman}.}
  \bibinfo{year}{2001}\natexlab{}.
\newblock \showarticletitle{Greedy Function Approximation: A Gradient Boosting
  Machine}.
\newblock \bibinfo{journal}{\emph{The Annals of Statistics}}
  \bibinfo{volume}{29}, \bibinfo{number}{5} (\bibinfo{year}{2001}),
  \bibinfo{pages}{1189--1232}.
\newblock
\href{https://doi.org/10.1214/aos/1013203451}{doi:\nolinkurl{10.1214/aos/1013203451}}


\bibitem[Guo et~al\mbox{.}(2017)]%
        {Guo2017Calibration}
\bibfield{author}{\bibinfo{person}{Chuan Guo}, \bibinfo{person}{Geoff Pleiss},
  \bibinfo{person}{Yu Sun}, {and} \bibinfo{person}{Kilian~Q. Weinberger}.}
  \bibinfo{year}{2017}\natexlab{}.
\newblock \showarticletitle{On Calibration of Modern Neural Networks}. In
  \bibinfo{booktitle}{\emph{Proceedings of the 34th International Conference on
  Machine Learning (ICML)}} \emph{(\bibinfo{series}{Proceedings of Machine
  Learning Research}, Vol.~\bibinfo{volume}{70})}. \bibinfo{pages}{1321--1330}.
\newblock


\bibitem[Hu et~al\mbox{.}(2020)]%
        {Hu2020TIFUKNN}
\bibfield{author}{\bibinfo{person}{Haoji Hu}, \bibinfo{person}{Xiangnan He},
  \bibinfo{person}{Jinyang Gao}, {and} \bibinfo{person}{Zhi-Li Zhang}.}
  \bibinfo{year}{2020}\natexlab{}.
\newblock \showarticletitle{Modeling Personalized Item Frequency Information
  for Next-basket Recommendation}. In \bibinfo{booktitle}{\emph{Proceedings of
  the 43rd International ACM SIGIR Conference on Research and Development in
  Information Retrieval (SIGIR '20)}}. \bibinfo{publisher}{ACM},
  \bibinfo{address}{Virtual Event, China}, \bibinfo{pages}{1071--1080}.
\newblock
\href{https://doi.org/10.1145/3397271.3401066}{doi:\nolinkurl{10.1145/3397271.3401066}}


\bibitem[{Instacart}(2017)]%
        {Instacart2017}
\bibfield{author}{\bibinfo{person}{{Instacart}}.}
  \bibinfo{year}{2017}\natexlab{}.
\newblock \bibinfo{title}{The {Instacart} Online Grocery Shopping Dataset
  2017}.
\newblock \bibinfo{howpublished}{Kaggle competition ``Instacart Market Basket
  Analysis'', \url{https://www.kaggle.com/c/instacart-market-basket-analysis}}.
\newblock


\bibitem[Jain et~al\mbox{.}(2016)]%
        {Jain2016XMLLoss}
\bibfield{author}{\bibinfo{person}{Himanshu Jain}, \bibinfo{person}{Yashoteja
  Prabhu}, {and} \bibinfo{person}{Manik Varma}.}
  \bibinfo{year}{2016}\natexlab{}.
\newblock \showarticletitle{Extreme Multi-label Loss Functions for
  Recommendation, Tagging, Ranking \& Other Missing Label Applications}. In
  \bibinfo{booktitle}{\emph{Proceedings of the 22nd ACM SIGKDD International
  Conference on Knowledge Discovery and Data Mining (KDD)}}.
  \bibinfo{pages}{935--944}.
\newblock
\href{https://doi.org/10.1145/2939672.2939756}{doi:\nolinkurl{10.1145/2939672.2939756}}


\bibitem[{Kaggle and Banco Santander}(2016)]%
        {KaggleSantander2016}
\bibfield{author}{\bibinfo{person}{{Kaggle and Banco Santander}}.}
  \bibinfo{year}{2016}\natexlab{}.
\newblock \bibinfo{title}{Santander Product Recommendation}.
\newblock
  \bibinfo{howpublished}{\url{https://www.kaggle.com/c/santander-product-recommendation}}.
\newblock


\bibitem[Kang et~al\mbox{.}(2020)]%
        {Kang2020Decoupling}
\bibfield{author}{\bibinfo{person}{Bingyi Kang}, \bibinfo{person}{Saining Xie},
  \bibinfo{person}{Marcus Rohrbach}, \bibinfo{person}{Zhicheng Yan},
  \bibinfo{person}{Albert Gordo}, \bibinfo{person}{Jiashi Feng}, {and}
  \bibinfo{person}{Yannis Kalantidis}.} \bibinfo{year}{2020}\natexlab{}.
\newblock \showarticletitle{Decoupling Representation and Classifier for
  Long-Tailed Recognition}. In \bibinfo{booktitle}{\emph{International
  Conference on Learning Representations (ICLR)}}.
\newblock
\urldef\tempurl%
\url{https://openreview.net/forum?id=r1gRTCVFvB}
\showURL{%
\tempurl}


\bibitem[Ke et~al\mbox{.}(2017)]%
        {Ke2017LightGBM}
\bibfield{author}{\bibinfo{person}{Guolin Ke}, \bibinfo{person}{Qi Meng},
  \bibinfo{person}{Thomas Finley}, \bibinfo{person}{Taifeng Wang},
  \bibinfo{person}{Wei Chen}, \bibinfo{person}{Weidong Ma},
  \bibinfo{person}{Qiwei Ye}, {and} \bibinfo{person}{Tie-Yan Liu}.}
  \bibinfo{year}{2017}\natexlab{}.
\newblock \showarticletitle{LightGBM: A Highly Efficient Gradient Boosting
  Decision Tree}. In \bibinfo{booktitle}{\emph{Advances in Neural Information
  Processing Systems 30 (NIPS)}}.
\newblock


\bibitem[Koyejo et~al\mbox{.}(2015)]%
        {Koyejo2015Consistent}
\bibfield{author}{\bibinfo{person}{Oluwasanmi~O. Koyejo},
  \bibinfo{person}{Nagarajan Natarajan}, \bibinfo{person}{Pradeep~K.
  Ravikumar}, {and} \bibinfo{person}{Inderjit~S. Dhillon}.}
  \bibinfo{year}{2015}\natexlab{}.
\newblock \showarticletitle{Consistent Multilabel Classification}. In
  \bibinfo{booktitle}{\emph{Advances in Neural Information Processing Systems
  28 (NIPS)}}.
\newblock


\bibitem[Kull et~al\mbox{.}(2017)]%
        {Kull2017Beta}
\bibfield{author}{\bibinfo{person}{Meelis Kull}, \bibinfo{person}{Telmo~M.
  Silva~Filho}, {and} \bibinfo{person}{Peter Flach}.}
  \bibinfo{year}{2017}\natexlab{}.
\newblock \showarticletitle{Beta Calibration: A Well-Founded and Easily
  Implemented Improvement on Logistic Calibration for Binary Classifiers}. In
  \bibinfo{booktitle}{\emph{Proceedings of the 20th International Conference on
  Artificial Intelligence and Statistics (AISTATS)}}
  \emph{(\bibinfo{series}{Proceedings of Machine Learning Research},
  Vol.~\bibinfo{volume}{54})}. \bibinfo{pages}{623--631}.
\newblock


\bibitem[Li et~al\mbox{.}(2023)]%
        {Li2023NBRRealityCheck}
\bibfield{author}{\bibinfo{person}{Ming Li}, \bibinfo{person}{Sami Jullien},
  \bibinfo{person}{Mozhdeh Ariannezhad}, {and} \bibinfo{person}{Maarten de
  Rijke}.} \bibinfo{year}{2023}\natexlab{}.
\newblock \showarticletitle{A Next Basket Recommendation Reality Check}.
\newblock \bibinfo{journal}{\emph{ACM Transactions on Information Systems}}
  \bibinfo{volume}{41}, \bibinfo{number}{4} (\bibinfo{year}{2023}),
  \bibinfo{pages}{116:1--116:29}.
\newblock
\href{https://doi.org/10.1145/3587153}{doi:\nolinkurl{10.1145/3587153}}


\bibitem[Lipton et~al\mbox{.}(2018)]%
        {Lipton2018LabelShift}
\bibfield{author}{\bibinfo{person}{Zachary~C. Lipton},
  \bibinfo{person}{Yu-Xiang Wang}, {and} \bibinfo{person}{Alexander Smola}.}
  \bibinfo{year}{2018}\natexlab{}.
\newblock \showarticletitle{Detecting and Correcting for Label Shift with Black
  Box Predictors}. In \bibinfo{booktitle}{\emph{Proceedings of the 35th
  International Conference on Machine Learning (ICML)}}
  \emph{(\bibinfo{series}{Proceedings of Machine Learning Research},
  Vol.~\bibinfo{volume}{80})}. \bibinfo{publisher}{PMLR},
  \bibinfo{pages}{3122--3130}.
\newblock
\urldef\tempurl%
\url{https://proceedings.mlr.press/v80/lipton18a.html}
\showURL{%
\tempurl}


\bibitem[Liu(2026)]%
        {Liu2026Resampling}
\bibfield{author}{\bibinfo{person}{Zewen Liu}.}
  \bibinfo{year}{2026}\natexlab{}.
\newblock \bibinfo{title}{The Hidden Cost of Resampling: How Imbalance
  Correction Degrades Probability Calibration in Tree Ensembles}.
\newblock \bibinfo{howpublished}{arXiv:2606.29720}.
\newblock


\bibitem[Masnadi-Shirazi and Vasconcelos(2011)]%
        {MasnadiShirazi2011CSBoost}
\bibfield{author}{\bibinfo{person}{Hamed Masnadi-Shirazi} {and}
  \bibinfo{person}{Nuno Vasconcelos}.} \bibinfo{year}{2011}\natexlab{}.
\newblock \showarticletitle{Cost-Sensitive Boosting}.
\newblock \bibinfo{journal}{\emph{IEEE Transactions on Pattern Analysis and
  Machine Intelligence}} \bibinfo{volume}{33}, \bibinfo{number}{2}
  (\bibinfo{year}{2011}), \bibinfo{pages}{294--309}.
\newblock
\href{https://doi.org/10.1109/TPAMI.2010.71}{doi:\nolinkurl{10.1109/TPAMI.2010.71}}


\bibitem[Matsubara(2024)]%
        {Matsubara2024WGBoost}
\bibfield{author}{\bibinfo{person}{Takuo Matsubara}.}
  \bibinfo{year}{2024}\natexlab{}.
\newblock \showarticletitle{Wasserstein Gradient Boosting: A Framework for
  Distribution-Valued Supervised Learning}. In
  \bibinfo{booktitle}{\emph{Advances in Neural Information Processing Systems
  37 (NeurIPS)}}.
\newblock


\bibitem[Menon et~al\mbox{.}(2021)]%
        {Menon2021LogitAdjust}
\bibfield{author}{\bibinfo{person}{Aditya~Krishna Menon},
  \bibinfo{person}{Sadeep Jayasumana}, \bibinfo{person}{Ankit~Singh Rawat},
  \bibinfo{person}{Himanshu Jain}, \bibinfo{person}{Andreas Veit}, {and}
  \bibinfo{person}{Sanjiv Kumar}.} \bibinfo{year}{2021}\natexlab{}.
\newblock \showarticletitle{Long-tail Learning via Logit Adjustment}. In
  \bibinfo{booktitle}{\emph{International Conference on Learning
  Representations (ICLR)}}.
\newblock


\bibitem[Menon et~al\mbox{.}(2019)]%
        {Menon2019Reductions}
\bibfield{author}{\bibinfo{person}{Aditya~Krishna Menon},
  \bibinfo{person}{Ankit~Singh Rawat}, \bibinfo{person}{Sashank~J. Reddi},
  {and} \bibinfo{person}{Sanjiv Kumar}.} \bibinfo{year}{2019}\natexlab{}.
\newblock \showarticletitle{Multilabel Reductions: What is My Loss
  Optimising?}. In \bibinfo{booktitle}{\emph{Advances in Neural Information
  Processing Systems 32 (NeurIPS)}}.
\newblock


\bibitem[Niculescu-Mizil and Caruana(2005)]%
        {NiculescuMizil2005}
\bibfield{author}{\bibinfo{person}{Alexandru Niculescu-Mizil} {and}
  \bibinfo{person}{Rich Caruana}.} \bibinfo{year}{2005}\natexlab{}.
\newblock \showarticletitle{Predicting Good Probabilities with Supervised
  Learning}. In \bibinfo{booktitle}{\emph{Proceedings of the 22nd International
  Conference on Machine Learning (ICML)}}. \bibinfo{pages}{625--632}.
\newblock
\href{https://doi.org/10.1145/1102351.1102430}{doi:\nolinkurl{10.1145/1102351.1102430}}


\bibitem[Nikolaou et~al\mbox{.}(2016)]%
        {Nikolaou2016CSBoost}
\bibfield{author}{\bibinfo{person}{Nikolaos Nikolaou},
  \bibinfo{person}{Narayanan Edakunni}, \bibinfo{person}{Meelis Kull},
  \bibinfo{person}{Peter Flach}, {and} \bibinfo{person}{Gavin Brown}.}
  \bibinfo{year}{2016}\natexlab{}.
\newblock \showarticletitle{Cost-Sensitive Boosting Algorithms: Do We Really
  Need Them?}
\newblock \bibinfo{journal}{\emph{Machine Learning}} \bibinfo{volume}{104},
  \bibinfo{number}{2--3} (\bibinfo{year}{2016}), \bibinfo{pages}{359--384}.
\newblock
\href{https://doi.org/10.1007/s10994-016-5572-x}{doi:\nolinkurl{10.1007/s10994-016-5572-x}}


\bibitem[Patel et~al\mbox{.}(2021)]%
        {Patel2021IMax}
\bibfield{author}{\bibinfo{person}{Kanil Patel}, \bibinfo{person}{William
  Beluch}, \bibinfo{person}{Bin Yang}, \bibinfo{person}{Michael Pfeiffer},
  {and} \bibinfo{person}{Dan Zhang}.} \bibinfo{year}{2021}\natexlab{}.
\newblock \showarticletitle{Multi-Class Uncertainty Calibration via Mutual
  Information Maximization-based Binning}. In
  \bibinfo{booktitle}{\emph{International Conference on Learning
  Representations (ICLR)}}.
\newblock


\bibitem[Platt(1999)]%
        {Platt1999}
\bibfield{author}{\bibinfo{person}{John~C. Platt}.}
  \bibinfo{year}{1999}\natexlab{}.
\newblock \showarticletitle{Probabilistic Outputs for Support Vector Machines
  and Comparisons to Regularized Likelihood Methods}.
\newblock In \bibinfo{booktitle}{\emph{Advances in Large Margin Classifiers}}.
  \bibinfo{publisher}{MIT Press}, \bibinfo{pages}{61--74}.
\newblock


\bibitem[Rahimi et~al\mbox{.}(2020)]%
        {Rahimi2020IntraOrder}
\bibfield{author}{\bibinfo{person}{Amir Rahimi}, \bibinfo{person}{Amirreza
  Shaban}, \bibinfo{person}{Ching-An Cheng}, \bibinfo{person}{Richard Hartley},
  {and} \bibinfo{person}{Byron Boots}.} \bibinfo{year}{2020}\natexlab{}.
\newblock \showarticletitle{Intra Order-preserving Functions for Calibration of
  Multi-Class Neural Networks}. In \bibinfo{booktitle}{\emph{Advances in Neural
  Information Processing Systems 33 (NeurIPS)}}.
\newblock


\bibitem[Saerens et~al\mbox{.}(2002)]%
        {Saerens2002}
\bibfield{author}{\bibinfo{person}{Marco Saerens}, \bibinfo{person}{Patrice
  Latinne}, {and} \bibinfo{person}{Christine Decaestecker}.}
  \bibinfo{year}{2002}\natexlab{}.
\newblock \showarticletitle{Adjusting the Outputs of a Classifier to New a
  Priori Probabilities: A Simple Procedure}.
\newblock \bibinfo{journal}{\emph{Neural Computation}} \bibinfo{volume}{14},
  \bibinfo{number}{1} (\bibinfo{year}{2002}), \bibinfo{pages}{21--41}.
\newblock
\href{https://doi.org/10.1162/089976602753284446}{doi:\nolinkurl{10.1162/089976602753284446}}


\bibitem[Sato(2024)]%
        {Sato2024TopN}
\bibfield{author}{\bibinfo{person}{Masahiro Sato}.}
  \bibinfo{year}{2024}\natexlab{}.
\newblock \showarticletitle{Calibrating the Predictions for Top-N
  Recommendations}. In \bibinfo{booktitle}{\emph{Proceedings of the 18th ACM
  Conference on Recommender Systems (RecSys)}}. \bibinfo{pages}{963--968}.
\newblock
\href{https://doi.org/10.1145/3640457.3688177}{doi:\nolinkurl{10.1145/3640457.3688177}}


\bibitem[Tae and Lee(2026)]%
        {Tae2026MRP}
\bibfield{author}{\bibinfo{person}{Inwoo Tae} {and} \bibinfo{person}{Yongjae
  Lee}.} \bibinfo{year}{2026}\natexlab{}.
\newblock \showarticletitle{Post-Calibration Reliability Reranking of Relevance
  Decisions via Label-wise Monotone Projection}. In
  \bibinfo{booktitle}{\emph{Proceedings of the 35th ACM International
  Conference on Information and Knowledge Management (CIKM)}}.
\newblock
\newblock
\shownote{arXiv:2608.10406}.


\bibitem[Thies et~al\mbox{.}(2026)]%
        {Thies2026LabelRankingCalib}
\bibfield{author}{\bibinfo{person}{Santo M. A.~R. Thies},
  \bibinfo{person}{Viktor Bengs}, \bibinfo{person}{Timo Kaufmann},
  \bibinfo{person}{Sebastian~J. Vollmer}, {and} \bibinfo{person}{Eyke
  H{\"u}llermeier}.} \bibinfo{year}{2026}\natexlab{}.
\newblock \bibinfo{title}{Calibrated Preference Learning: The Case of Label
  Ranking}.
\newblock \bibinfo{howpublished}{arXiv:2605.30447}.
\newblock


\bibitem[Tsoumakas et~al\mbox{.}(2011)]%
        {Tsoumakas2011MULAN}
\bibfield{author}{\bibinfo{person}{Grigorios Tsoumakas},
  \bibinfo{person}{Eleftherios Spyromitros-Xioufis}, \bibinfo{person}{Jozef
  Vilcek}, {and} \bibinfo{person}{Ioannis Vlahavas}.}
  \bibinfo{year}{2011}\natexlab{}.
\newblock \showarticletitle{MULAN: A Java Library for Multi-Label Learning}.
\newblock \bibinfo{journal}{\emph{Journal of Machine Learning Research}}
  \bibinfo{volume}{12} (\bibinfo{year}{2011}), \bibinfo{pages}{2411--2414}.
\newblock


\bibitem[Ullah et~al\mbox{.}(2025)]%
        {Ullah2025XMLCCalib}
\bibfield{author}{\bibinfo{person}{Nasib Ullah}, \bibinfo{person}{Erik
  Schultheis}, \bibinfo{person}{Jinbin Zhang}, {and} \bibinfo{person}{Rohit
  Babbar}.} \bibinfo{year}{2025}\natexlab{}.
\newblock \showarticletitle{How Well Calibrated are Extreme Multi-label
  Classifiers? An Empirical Analysis}. In \bibinfo{booktitle}{\emph{Proceedings
  of the 31st ACM SIGKDD Conference on Knowledge Discovery and Data Mining
  (KDD)}}. \bibinfo{pages}{1397--1408}.
\newblock
\href{https://doi.org/10.1145/3690624.3709333}{doi:\nolinkurl{10.1145/3690624.3709333}}


\bibitem[van~den Goorbergh et~al\mbox{.}(2022)]%
        {VandenGoorbergh2022}
\bibfield{author}{\bibinfo{person}{Ruben van~den Goorbergh},
  \bibinfo{person}{Maarten van Smeden}, \bibinfo{person}{Dirk Timmerman}, {and}
  \bibinfo{person}{Ben Van~Calster}.} \bibinfo{year}{2022}\natexlab{}.
\newblock \showarticletitle{The Harm of Class Imbalance Corrections for Risk
  Prediction Models: Illustration and Simulation Using Logistic Regression}.
\newblock \bibinfo{journal}{\emph{Journal of the American Medical Informatics
  Association}} \bibinfo{volume}{29}, \bibinfo{number}{9}
  (\bibinfo{year}{2022}), \bibinfo{pages}{1525--1534}.
\newblock
\href{https://doi.org/10.1093/jamia/ocac093}{doi:\nolinkurl{10.1093/jamia/ocac093}}


\bibitem[Wallace and Dahabreh(2012)]%
        {Wallace2012}
\bibfield{author}{\bibinfo{person}{Byron~C. Wallace} {and}
  \bibinfo{person}{Issa~J. Dahabreh}.} \bibinfo{year}{2012}\natexlab{}.
\newblock \showarticletitle{Class Probability Estimates are Unreliable for
  Imbalanced Data (and How to Fix Them)}. In \bibinfo{booktitle}{\emph{2012
  IEEE 12th International Conference on Data Mining (ICDM)}}.
  \bibinfo{pages}{695--704}.
\newblock
\href{https://doi.org/10.1109/ICDM.2012.115}{doi:\nolinkurl{10.1109/ICDM.2012.115}}


\bibitem[Wang et~al\mbox{.}(2026)]%
        {Wang2026BeyondLA}
\bibfield{author}{\bibinfo{person}{Zhanliang Wang}, \bibinfo{person}{Hongzhuo
  Chen}, \bibinfo{person}{Quan~Minh Nguyen}, \bibinfo{person}{Mian~Umair
  Ahsan}, {and} \bibinfo{person}{Kai Wang}.} \bibinfo{year}{2026}\natexlab{}.
\newblock \showarticletitle{Beyond Logit Adjustment: A Residual Decomposition
  Framework for Long-Tailed Reranking}. In \bibinfo{booktitle}{\emph{Third
  Conference on Language Modeling (COLM)}}.
\newblock
\newblock
\shownote{arXiv:2604.01506}.


\bibitem[Wydmuch et~al\mbox{.}(2018)]%
        {Wydmuch2018PLT}
\bibfield{author}{\bibinfo{person}{Marek Wydmuch}, \bibinfo{person}{Kalina
  Jasinska}, \bibinfo{person}{Mikhail Kuznetsov}, \bibinfo{person}{R{\'o}bert
  Busa-Fekete}, {and} \bibinfo{person}{Krzysztof Dembczy{\'n}ski}.}
  \bibinfo{year}{2018}\natexlab{}.
\newblock \showarticletitle{A No-Regret Generalization of Hierarchical Softmax
  to Extreme Multi-Label Classification}. In \bibinfo{booktitle}{\emph{Advances
  in Neural Information Processing Systems 31 (NeurIPS)}}.
\newblock


\bibitem[{XGBoost developers}(2026)]%
        {XGBoostDocs}
\bibfield{author}{\bibinfo{person}{{XGBoost developers}}.}
  \bibinfo{year}{2026}\natexlab{}.
\newblock \bibinfo{title}{{XGBoost} Parameters and Notes on Parameter Tuning}.
\newblock
  \bibinfo{howpublished}{\url{https://xgboost.readthedocs.io/en/stable/parameter.html}
  and
  \url{https://xgboost.readthedocs.io/en/stable/tutorials/param_tuning.html}}.
\newblock
\newblock
\shownote{Accessed 2026-09-07}.


\bibitem[Yu et~al\mbox{.}(2016)]%
        {Yu2016DREAM}
\bibfield{author}{\bibinfo{person}{Feng Yu}, \bibinfo{person}{Qiang Liu},
  \bibinfo{person}{Shu Wu}, \bibinfo{person}{Liang Wang}, {and}
  \bibinfo{person}{Tieniu Tan}.} \bibinfo{year}{2016}\natexlab{}.
\newblock \showarticletitle{A Dynamic Recurrent Model for Next Basket
  Recommendation}. In \bibinfo{booktitle}{\emph{Proceedings of the 39th
  International ACM SIGIR Conference on Research and Development in Information
  Retrieval (SIGIR '16)}}. \bibinfo{publisher}{ACM}, \bibinfo{address}{Pisa,
  Italy}, \bibinfo{pages}{729--732}.
\newblock
\href{https://doi.org/10.1145/2911451.2914683}{doi:\nolinkurl{10.1145/2911451.2914683}}


\bibitem[Zadrozny and Elkan(2002)]%
        {Zadrozny2002}
\bibfield{author}{\bibinfo{person}{Bianca Zadrozny} {and}
  \bibinfo{person}{Charles Elkan}.} \bibinfo{year}{2002}\natexlab{}.
\newblock \showarticletitle{Transforming Classifier Scores into Accurate
  Multiclass Probability Estimates}. In \bibinfo{booktitle}{\emph{Proceedings
  of the Eighth ACM SIGKDD International Conference on Knowledge Discovery and
  Data Mining (KDD)}}. \bibinfo{pages}{694--699}.
\newblock
\href{https://doi.org/10.1145/775047.775151}{doi:\nolinkurl{10.1145/775047.775151}}


\end{thebibliography}

\end{document}